\documentclass[journal,onecolumn]{IEEEtran}

\usepackage[pdftex]{graphicx} 

\usepackage{calc} 
\usepackage{enumitem} 

\usepackage[a4paper, lmargin=1.5cm, rmargin=1.5cm, tmargin=3.35cm, bmargin=2.83cm]{geometry} 

\usepackage[all]{nowidow} 
\usepackage[protrusion=true,expansion=true]{microtype} 

\usepackage{xspace}

\makeatletter
\DeclareRobustCommand\onedot{\futurelet\@let@token\@onedot}
\def\@onedot{\ifx\@let@token.\else.\null\fi\xspace}

\def\ie{i.e\onedot} 
\def\cf{cf\onedot}

\def\supp{\text{supp}}
\def\pr{{\rm{Pr}}}
\makeatother

\usepackage{amsmath, nccmath}
\usepackage{amssymb,amsthm}
\usepackage{amsfonts}
\usepackage{cite}
\usepackage{tikz}
\usepackage{subcaption}

\usepackage{overpic}
\usepackage{float}
\usepackage{bm}
\usepackage{bbm}

\usepackage{pgfplots}
\pgfplotsset{compat=newest,compat/show suggested version=false}

\usepackage{booktabs}
\usepackage{colortbl}
\usepackage{color}
\usepackage{blkarray, bigstrut}
\usepackage{mathtools}
\usepackage{proof-at-the-end}
\usepackage{dsfont}
\DeclarePairedDelimiter{\floor}{\lfloor}{\rfloor}
\DeclarePairedDelimiter{\ceil}{\lceil}{\rceil}

\theoremstyle{plain}
\newtheorem{theorem}{Theorem}
\newtheorem{lemma}{Lemma}
\newtheorem{corollary}{Corollary}
\newtheorem{proposition}{Proposition}

\theoremstyle{definition}

\newtheorem{example}{Example}
\theoremstyle{remark}
\newtheorem{remark}{Remark}

\newcommand{\mc}[1]{\mathcal{#1}}
\newcommand{\mb}[1]{\mathbf{#1}}
\newcommand{\bb}[1]{\mathbb{#1}}

\newcommand{\utag}[2]{\mathop{#2}\limits^{\text{(#1)}}}
\newcommand{\uref}[1]{(#1)}

\allowdisplaybreaks

\usepackage{setspace}
\begin{document}

\title{Lossy Compression for Sparse Aggregation}


\author{
    Yijun Fan, Fangwei~Ye, and Raymond~W.~Yeung, 
\thanks{
This paper was presented in part in "Lossy Compression for Sparse Aggregation,"  IEEE Information Theory Workshop, Shenzhen, China, 2024.

Y.~Fan and R.~W.~Yeung are with the Department of Information Engineering, The Chinese University of Hong Kong, Hong Kong (emails: \{yijunfan, whyeung\}@ie.cuhk.edu.hk).

F.~Ye is with the College of Computer Science and Technology, Nanjing University of Aeronautics and Astronautics, Nanjing 211106, China (email: fangweiye@nuaa.edu.cn).

}
}

\maketitle

\begin{abstract}
We consider the problem of transmitting sparse local updates to the server in a distributed learning system. Specifically,  
the system consists of $n$ clients, each possessing a $k$-sparse $d$-dimensional local model, and a central server responsible for aggregating the clients' models into a global model. 
The goal is to characterize the tradeoff between the communication cost in the transmission from the clients to the server and the accuracy in aggregating the global model. 

We propose a compression scheme for sparse local models by concatenating a covering method and a sketching method. We also present a converse based on f-divergence, which strengthens the conventional Fano-type lower bounds. The proposed lower bound is tight for the frequency estimation case, that is, each coordinate takes values in a binary alphabet.  
For general alphabets, the proposed achievable schemes remain suboptimal relative to the converse bounds, indicating that a complete characterization of the communication-accuracy tradeoff requires further investigation.
\end{abstract}

\section{Introduction}

The rapid growth of data size and model dimension has made communication a central bottleneck in distributed systems. In federated learning \cite{kim2021federated,kairouz2021advances,vargaftik2022eden}, distributed SGD \cite{gandikota2022vqsgd,mcmahan2017communication}, and sensor networks for the Internet of Things (IoT) \cite{sarkar2014diat}, many clients repeatedly transmit high-dimensional local data to a central server for aggregation. When communication resources are limited, exact transmission may be infeasible, and it becomes important to design compressed representations that preserve the accuracy of the aggregate. This paper focuses on sparse data, which arises in communication-constrained federated learning \cite{strom2015scalable,aji2017sparse,sattler2019sparse,li2021communication,rothchild2020fetchsgd,tsuzuku2018variance,amiri2020federated}, submodel learning \cite{niu2019secure,xu2021deepreduce}, heavy-hitter estimation \cite{acharya2019communication,zhou2022locally}, and histogram estimation \cite{chen2023communication}. Under these scenarios, transmitting sparse data entry-wise would incur substantial communication overhead, motivating a good representation of sparse data in
high dimensions.

The communication-accuracy tradeoff for distributed estimation has been studied from several perspectives. Classical rate-distortion theory characterizes fundamental rates under distortion constraints \cite{slepian1973noiseless}, while recent non-asymptotic works study communication-constrained parameter estimation \cite{zhang2013information,kipnis2022mean} and distribution learning \cite{han2018distributed,acharya2020inference,huang2022collaborative,nam2023optimal}. Distributed mean estimation for high-dimensional vectors has also been investigated using vector quantization \cite{chen2022breaking,isik2024exact} and entry-wise quantization \cite{suresh2017distributed,reisizadeh2020fedpaq,haddadpour2021federated}, including schemes without prior knowledge of the local vector distribution \cite{suresh2017distributed} and extensions to correlated vectors \cite{suresh2022correlated}. These methods are well suited to dense vectors, where the optimal communication cost typically scales linearly with the dimension. For sparse vectors, however, the relevant complexity should depend on the sparsity and the allowed distortion, rather than only on the full dimension.

\begin{figure}
    \centering
    {\includegraphics[width=0.4\linewidth]{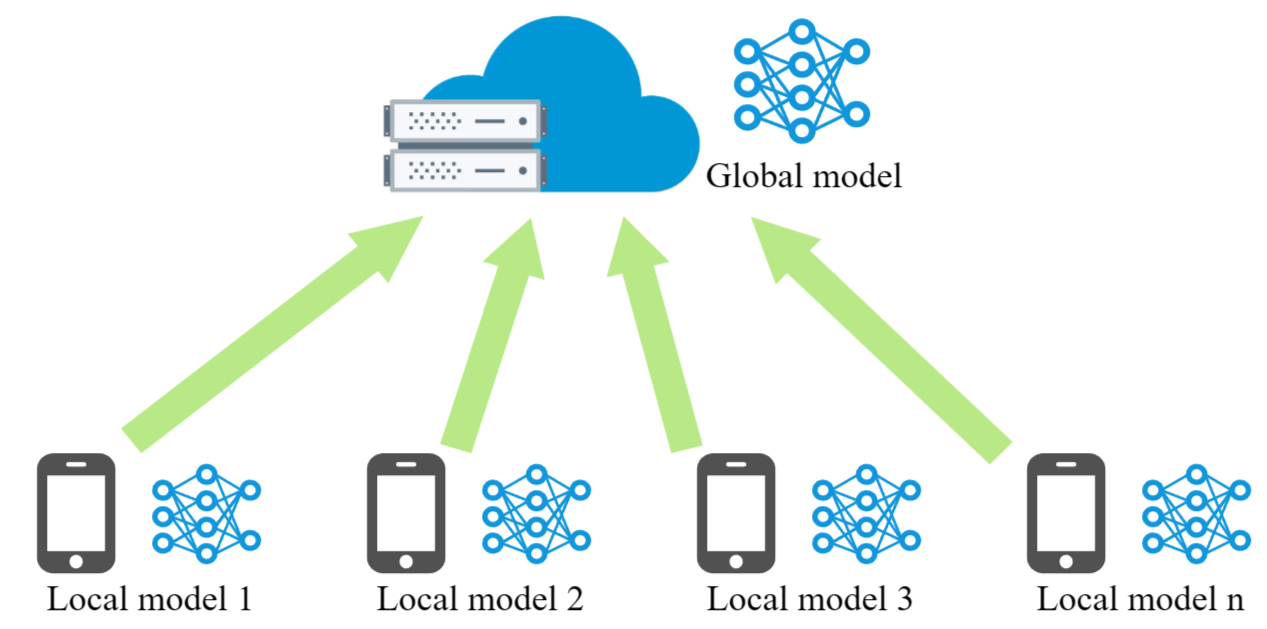}}    
    \caption{Overview of a distributed system in federated learning consisting of $n$ clients and a server. The global model will then be broadcast to all clients. We only consider the uplinks as demonstrated in this picture.}
\label{fig::fl}
\end{figure}

Several compression methods have been used for sparse aggregation. A direct approach is to transmit the indices and values of nonzero entries, as in sparse gradient compression \cite{strom2015scalable} and subsequent federated learning works \cite{sattler2019sparse,tsuzuku2018variance,niu2019secure,amiri2020federated,xu2021deepreduce}. Another line of work uses sketching methods from the data-streaming literature, such as Count Sketch \cite{charikar2002finding} and Count-Min Sketch \cite{cormode2005improved}, which are naturally suited to frequency-type estimation and distributed aggregation. Sketching has also been studied under security and privacy considerations \cite{chen2023communication,liu2019enhancing,zhou2022locally}. From a coding perspective, algebraic constructions based on Reed-Solomon codes provide deterministic sparse recovery guarantees \cite{das2013finite} and have been applied to distributed sparse aggregation \cite{pan2022machine}. These approaches motivate the present work, but most existing guarantees focus either on exact recovery or on sketching performance without explicitly characterizing the lossy sparse aggregation tradeoff.

In this paper, we study sparse aggregation under an $\ell_1$ distortion and characterize the communication-accuracy tradeoff from both achievability and converse perspectives. On the achievability side, we propose a covering-based compression scheme showing that distortion reduces the effective sparsity from $k$ to approximately $k-D/q$, and we further combine this scheme with Count Sketch and Count-Min Sketch to obtain improved randomized compression schemes. On the converse side, we develop an f-divergence-based lower bound that strengthens classical Fano-type arguments used in communication-constrained estimation \cite{zhang2013information,chen2023communication,zhou2022locally}. In the binary case, corresponding to distributed frequency estimation, the achievable and converse bounds match up to lower-order terms. For general alphabets, the full tradeoff is not yet completely characterized due to the gap remained.

This paper is organized as follows. Section~\ref{sec::formulation} presents the system model and the communication-accuracy tradeoff considered in this work. Section~\ref{sec::related_works} reviews the necessary background on compressive sensing and sketching methods. Section~\ref{sec::main_results} previews the main contributions of this paper. Section~\ref{sec::achievability} provides the achievability results, including the covering-based compression scheme, the analyses of Count Sketch and Count-Min Sketch, and their combination with the covering scheme. Section~\ref{sec::converse} proves converse bounds based on an f-divergence information measure and compares them with the achievable communication costs. We conclude the paper and outline future directions in Section~\ref{sec::discussion}.

Throughout this paper, the real number set, the integer set, and positive integer numbers are denoted by $\bb R$, $\bb Z$, and $\bb Z_{+}$, respectively.
We denote the $j$-th entry of a vector in brackets, such as $X(j)$. The estimate of $X$ is denoted by $\hat{X}$. 
We denote the support of a $d$-dimensional vector $X$ by $\text{supp}(X):=\{j\in\{1,2,\ldots,d\}: X(j)\neq 0\}$. The indicator function of some condition $A$ is denoted by $\mathds{1}(A)$, where $\mathds{1}(A)=1$ if and only if the condition $A$ is true and is $0$ otherwise. The $\ell_0$- and $\ell_1$-norm of $X$ are $\|X\|_0$ and $\| X\|_1$.
For any $q\in\bb Z_+$, $[q]$ represents the integer number set $\{1,\ldots,q\}$ and $[q_1:q_2]$ represents the set $\{q_1,q_1+1,\ldots,q_2\}$ for any $q_1\leq q_2$, $q_1,q_2\in\bb Z$. We denote the function $\max(x,0)$ by $(x)^+$. All logarithms throughout this paper are in the base $2$. 

\section{Problem Formulation}
\label{sec::formulation}


We consider a classical distributed learning setting with $n$ clients and a central server. Assume the clients are indexed by $1,\ldots,n$. In every round, each client will send its local model $ X_i$ to the server and the server will aggregate them to obtain the global model. The clients are not allowed to communicate with each other.

In this paper, the local model $ X_i$ is represented by a $d$-dimensional (discrete) vector, where each coordinate of the vector takes values in $[0:q]$. Then $ X_i\in[0:q]^d$ for all $i=1,2,\ldots,n$, where $$[0:q]^d:=\{ x\in \bb R^d: x(j)\in[0:q] \text{ for all }j=1,\ldots,d\}.$$ One may think of it as the local model after the quantization, where the vector (model) takes values in real numbers in general. We assume no prior knowledge about the distributions of local models.

The (normalized) aggregated global model is denoted by $ U$, given by,
\begin{equation}
\label{eq::def_u}
     U = \frac{1}{n} \sum_{i=1}^n  X_i.
\end{equation} 
The collection of all possible aggregated results $ U$ is denoted by $\mc U$.

We are particularly interested in the case that the local model is sparse, as in many applications such as \cite{strom2015scalable,aji2017sparse, sattler2019sparse,li2021communication, rothchild2020fetchsgd, tsuzuku2018variance, amiri2020federated,chen2023communication, zhou2022locally}. Specifically, each $d$-dimensional vector $ X_i$ is $k$-sparse, meaning that $ X_i$ for $i=1,\ldots,n$ takes values in the alphabet 
\begin{equation*}
    \mc X_{d,k,q}:=\{ x\in[0:q]^d:\| x\|_0\leq k\}.
\end{equation*}
For sparsity, we assume that $d \gg n k$, that is, both the local models $ X_i$'s and aggregated global model $U$ are sparse. 

For communication efficiency, each client may compress its sparse local model before transmission. Specifically, suppose that each client $i$ encodes its local model $ X_i$ by a common encoder $\mc E$ into a compressed model $\mc E(X_i)\in\mc Y$ to the server, where $\mc Y$ is some general finite alphabet. Since we make no assumption on the prior of local models, the number of bits required in transmitting the compressed model $Y_i$ is $\lceil \log |\mc Y|\rceil$.
After receiving the compressed local models, the server reproduces the aggregated global model \begin{equation}
    \label{eq::def_hat_u}
    \hat{U}=\mc D(\mc E( X_1),\ldots,\mc E( X_n))
\end{equation}
by a decoder $\mc D$. The system is fully described by the pair of encoding and decoding functions $\pi =(\mc E, \mc D)$. For simplicity, we write \[
\mc E(X^n) = (\mc E(X_1),\ldots,\mc E(X_n))
\] throughout this paper. Fig.~\ref{fig::model} is an illustration of the system.

We are interested in solving the problem of minimizing the communication cost under the reproduction accuracy constraint. 
The communication cost per client is measured by $\lceil \log |\mc Y|\rceil$, the number of bits transmitted per client.
The total communication cost of a scheme $\pi=(\mc E,\mc D)$ is then $n\lceil\log |\mc Y|\rceil$,
which is the total number of bits to transmit $n$ compressed local models $\mc E(X^n)$, denoted by $C(\pi)$. 

The accuracy of the reproduced aggregated global model is measured by $\ell(U, \hat{U})$, where $\ell(\cdot, \cdot)$ is the loss function and we consider $\ell_1$-loss $\| U-\hat{ U}\|_1$ throughout this paper. For a scheme $\pi$, we say the reproduced global model $\hat{ U}$ is $(D,\delta)$-accurate, $\delta\in(0,1)$, if and only if
\begin{equation}
\label{eq::accuracy}
    Pr(\|U - \hat{ U}\|_1 \leq  D) \geq 1-\delta.
\end{equation} 

We are interested in solving the problem 
\begin{equation}
\label{eq::def_cd}
    C(D,\delta) = \inf_{\pi\in\Pi(D,\delta)} C(\pi),
\end{equation}
where $\Pi(D,\delta)$ denote the collection of all schemes $\pi$ reproducing $(D,\delta)$-accurate $\hat{U}$ for any $n\geq 1$. For a fixed $\delta$, $C(D,\delta)$ can be regarded as the tradeoff between the communication cost and the distortion.

\begin{figure}
    \centering
    {\includegraphics[width=0.4\linewidth]{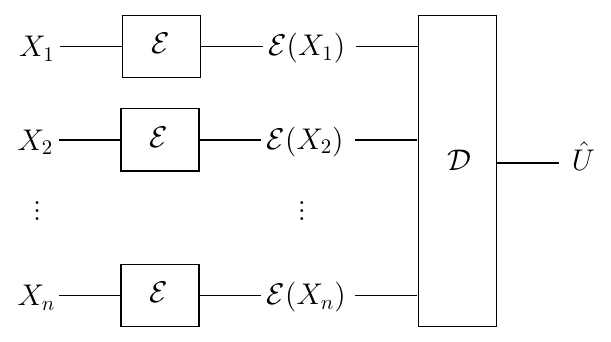}}    
    \caption{Each client $i$ encodes its sparse vector $ X_i$ into $M_i=\mc E( X_i)\in\mc Y$ and transmits it to the server end, which computes the aggregation result $\hat{ U}$ by $\mc D(\mc E( X_1),\ldots,\mc E( X_n))$.}
\label{fig::model}
\end{figure}

\begin{remark}
    When the distribution of local models is unknown, one cannot assume that local supports overlap in a favorable way. In the worst case, the nonzero coordinates across clients may be largely disjoint, so that the aggregate support is essentially the union of all local supports. As a result, the effective global sparsity can scale on the order of $nk$, indicating that the number of relevant global heavy hitters cannot be assumed to be significantly small. 
\end{remark}

\section{Preliminary and Background}
\label{sec::related_works}

\noindent\textbf{Compressive Sensing:}
Compressive sensing is a standard tool for compressing sparse vectors by replacing the original high-dimensional representation with a smaller number of linear measurements. For finite-alphabet sparse vectors, this idea can be implemented using coding constructions. In particular, Das and Vishwanath showed that when the measurement matrix is chosen according to the parity-check matrix of a suitable Reed-Solomon code, any $k$-sparse vector can be recovered exactly from approximately $2k$ measurements over a sufficiently large finite field~\cite{das2013finite}. Consequently, transmitting the measurements requires on the order of $2k\log(dq)$ bits per client.

This approach has also been used in distributed sparse aggregation, where each client sends linear measurements of its local model to the server~\cite{pan2022machine}. The main advantage of such schemes is that they provide deterministic, worst-case recovery guarantees while exploiting sparsity. However, these guarantees are designed for exact reconstruction. Therefore, the communication cost remains tied to the original sparsity $k$, and the resulting schemes do not directly characterize how the required communication decreases with distortion.
\vspace{1em}

\noindent\textbf{Sketching Methods:}
It is found in \cite{rothchild2020fetchsgd} \cite{ahn2024model} \cite{liu2019enhancing} 
that sketching methods, which are first proposed for data streaming, can be applied to federated learning to hide clients' identity and compress the local models. 
The work \cite{chen2023communication} considers binary $1$-sparse vectors using a modified version of Count Sketch (CS) with threshold. On the other hand, both \cite{zhou2022locally} and \cite{naim2022private} simplify CS and apply it to compress $k$-sparse vectors taking discrete finite values.  

Sketching methods offer two key advantages. First, their additive structure allows clients to compress data locally, while the server can directly sum the compressed messages to obtain a sketch of the aggregate. This is useful for identifying global patterns such as heavy hitters. Second, because sketches are randomized linear combinations rather than explicit coordinate transmissions, they implicitly hide local support information, making them naturally more privacy-friendly.

The advantage from the additive structure, however, is weakened under the unknown local-model distribution assumption. In the worst case, the supports of the local models may be essentially disjoint, so the sparsity of the aggregate can scale on the order of $nk$. As a result, the communication cost of sketching methods may grow significantly with $n$, limiting the benefit of directly aggregating the sketches.

\section{Main Results}
\label{sec::main_results}

This paper studies the communication-accuracy tradeoff for distributed aggregation of sparse vectors by providing improved achievable schemes and sharper converse bounds. The main structural message is that distortion reduces the number of coordinates that must be communicated accurately. Thus, the relevant parameter is not simply the original sparsity $k$, but an effective sparsity on the order of $k-D/q$, which determines the dominant communication cost.

On the achievability side, we first develop a deterministic covering-based compression scheme. For every local vector $x$, the encoder produces a compressed vector $y(x)$ such that
\[
    \|x-y(x)\|_1 \leq D .
\]
The construction keeps only the large entries of $x$, and every compressed vector satisfies
\[
    \|y(x)\|_0 \le \left\lceil k-D/q\right\rceil .
\]
Consequently, the number of messages required by the covering code obeys
\[
    \log |\mc Y(D)| \le \left\lceil k-D/q\right\rceil \log(dq) + \mc O(1),
\]
allowing an approximately linear decrease of the communication cost with the distortion. This improves the Reed--Solomon type schemes.

We then analyze sketching-based compression schemes, including Count Sketch and Count-Min Sketch, whose communication costs scale polynomially with the sparsity level and can be inefficient in worst-case sparse aggregation. To address this, we combine covering and sketching in a two-stage scheme. The total distortion is split as $D=D_1+D_2$: the covering stage uses $D_1$ to reduce the sparsity, while the sketching stage uses $D_2$ to compress the remaining vector. For example, when $D \in [mq,(m+1)q)$, the combined covering and Count-Min Sketch scheme yields a bound of the form
\[
    C(D,\delta) \lesssim
    \begin{cases}
        n\,\frac{k^2}{m}\log \frac{d}{\delta}, & \text{for small } m\\
        n\,(k-m)\log \frac{d}{\delta}, & \text{for large } m.
    \end{cases}
\]
This reveals a transition from a sketching-dominated regime at small distortion to an effective-sparsity-dominated regime at larger distortion.

On the converse side, we prove a lower bound based on f-divergence, which strengthens classical Fano-type arguments. For any $(D,\delta)$-accurate scheme with alphabet $\mathcal{Y}$, the bound relates $|\mathcal{Y}|$ to the mass $P_{\max}(D)$ of the largest distortion $\ell_1$ ball around the global model:
\begin{align*}
    \mathrm{f}(|\mathcal{Y}|^n) \gtrsim (1-P_{\max}(D)) \mathrm{f}\left(\frac{\delta}{1-P_{\max}(D)}\right)
    + P_{\max}(D)\mathrm{f}\left(\frac{1-\delta}{P_{\max}(D)}\right).
\end{align*}
By choosing suitable $f$, this framework yields tighter lower bounds than those obtained from KL-based Fano's inequality. In particular, for the binary case $q=1$, the achievable and converse bounds essentially match, giving
\[
    C(D,0) \approx n(k-D)\log d.
\]
Therefore, for frequency estimation, the communication-accuracy tradeoff is fully characterized up to lower-order terms, while for general alphabets the proposed bounds substantially narrow the gap between achievability and converse.

\section{Achievable Communication Cost}
\label{sec::achievability}


\subsection{Covering Compression}
\label{sec::covering}
We start by presenting the deterministic compression scheme, which is referred to as the covering scheme in this paper. The geometric interpretation of this scheme is to find a subset $\mc Y \subseteq [0:q]^{d}$ such that for every $ x\in \mc X_{d,k,q}$, we can map it to a $ y\in\mc Y$ satisfying $\| x- y\|_1\leq D$. That is, every local model vector $X$ could be covered by a $Y$ such that their distance $\|X-Y\|_1\leq D$. 
Noting that the communication cost increases with the size of the subset $\mc Y$, our goal is to design a $\mc Y$ with a small cardinality. In the rest of this paper, we will denote the subset constructed by the proposed covering scheme by $\mc Y(D)$ for any distortion $D\geq 0$.

A detailed description of our scheme is as follows. Given $D \geq 0$, for any vector $ x \in\mc X_{d,k,q}$, we try to map it to a vector $y$, whose $\ell_1$-distance from the original vector $x$ is below the threshold $D$ with probability $1$. 
To facilitate our discussion, we let $x(0)=\infty$ and refer to the vector $[x(0),x(1),\ldots,x(d)]$ as the extension of $x$.
Specifically, each client first reorders the extension of $x$, \ie, its local model $x$ together with $x(0)$, by a permutation $s$ on $[0:d]$ such that for any $i$ and $j$, 
\begin{align}
\label{eq::mapping_order}
s(i)<s(j) &\text{ if and only if }
\left\{ \begin{array}{ll}
     x(i)>x(j), \text{ or }\\
     i<j \text{ when } x(i)=x(j)
\end{array} \right.
\end{align} 
and let $\check{x}(s(i))=x(i)$ for $i=0,1,\ldots,d$. Note that $s$ is uniquely defined by \eqref{eq::mapping_order}, $s(0)=0$, and $\check{x}(0)=x(0)=\infty$.
For clarity, we will refer to the vector $[\check{x}(1),\check{x}(2),\ldots,\check{x}(d)]$ as $\check{x}$ and the vector $[\check{x}(0),\check{x}(1),\ldots,\check{x}(d)]$ as the extension of $\check{x}$.
In other words, $\check{x}$ is obtained by rearranging the coordinates of $x$ in decreasing order and preserving the original order if multiple coordinates share the same value. 
For the extension of $\check{x}$, let $s^*$ be the largest $r\geq 0$ satisfying \begin{equation}
\label{eq::mapping_rule}
    \check{x}(r)+\check{x}(r+1)+\ldots+\check{x}(d)> D,
\end{equation}
with $s^*=s(i^*)$, where $i^*\in[0:d]$. If \eqref{eq::mapping_rule} is not satisfied by any $r\in[d]$, then $s^*$ is not defined. Note that $i^*$ is uniquely determined by $s^*$.
Also, note that by virtue of \eqref{eq::mapping_rule}, we have $\check{x}(s^*)>0$.
Then map the local model $x$ to $ y(x)$, where, for $i\in[d]$, 
\begin{equation}
\label{eq::mapping_rule_2}
 y(x)(i) =
\left\{ \begin{array}{ll}
     \check{x}(s(i))=x(i) &  \text{ for } i\text{ s.t. } s(i)\leq s(i^*)\\
     0 &  \text{ for } i\text{ s.t. } s(i)\geq  s(i^*)+1.
\end{array} \right.
\end{equation}
See Example~\ref{example::mapping} below for an illustration.




For notational simplicity, in the rest of this section, given a vector $y\in[0:q]^d$ and $y\neq \mb 0$, we define \begin{align}
    \label{eq::min_value}m(y) &= \min\{y(i): i\in[d], y(i)>0\}\\
    \label{eq::set_a_index} i_m(y) &= \max\{i\in[d]:y(i)=m(y)\}\\
    \label{eq::set_a_def} \mc A(y) &= \{i\in[d]: i> i_m(y) \text{ and }y(i)=0\},
\end{align}
\ie, $m(y)$ is the value of the smallest non-zero entry in $y$, $i_m(y)$ is the index of the last occurrence of the smallest non-zero entry in $y$, and $\mc A(y)$ is the index collection of the zero entries in $y$ that are after the last occurrence of the value $m(y)$. 


We use the following example as an illustration of the mapping from $ x\in\mc X_{d,k,q}$ to $ y(x)$ under a given distortion $D$, as well as the quantities defined in \eqref{eq::min_value} to \eqref{eq::set_a_def}.
\begin{example}
\label{example::mapping}
    Suppose the client observes a local model \[
    x=(0,0,2,0,1,7,2,0,0).
    \]
    To satisfy \eqref{eq::mapping_order}, \eqref{eq::mapping_rule}, and \eqref{eq::mapping_rule_2}, the mapping $s$ is given by \begin{align*}
        s(6)=1, s(3)=2, s(7)=3, s(5)=4,s(1)=5, s(2)=6, s(4)=7, s(8)=8, s(9)=9,
    \end{align*}
    and the corresponding $\check{x}$ with $\check{x}(s(i))=x(i)$ is \[
    (7,2,2,1,0,0,0,0,0).
    \]
    
    Suppose $D=4$. Since \[
    \sum_{r=2}^9\check{x}(r)=5>D \ \ \ \text{and}\ \ \ \sum_{r=3}^9 \check{x}(r)=3<D,
    \]
    we have $s^*=2$ and $i^*=3$. Then in the compressed local model $y(x)$, $y(x)(i)=x(i)$ for $i$ such that $s(i)\leq s^*$ ($i=6,3$) and $y(x)(i)=0$ for $i$ such that $s(i)\geq s^*+1$ ($i=7,5,1,2,4,8,9$). Hence, \[
    y(x) = (0,0,2,0,0,7,0,0,0).
    \] 
    For this vector $y(x)$, \begin{align*}
        m(y(x)) = 2, \,i_m(y(x)) = 3, \text{ and } \mc A(y(x)) = \{4,5,7,8,9\}.
    \end{align*}

    Suppose $D=15$. Since \[
    \sum_{r=0}^9 \check{x}(r)=\infty>D\ \ \ \text{and}\ \ \ \sum_{r=1}^9 \check{x}(r)=12<D,
    \]
    we have $s^*=0$ and $i^*=0$. Hence, for all $i\in[1:d]$, $s(i)\geq s(i^*)+1$. By \eqref{eq::mapping_rule_2}, $y(x)(i)=0$ for all $i\in[1:d]$. That is, $y(x)=\mb 0$ and $m(y(x))$, $i_m(y(x))$, and $\mc A(y(x))$ are not defined in this case.
\end{example}
\vspace{1em}

\begin{remark}
\label{remark::D_le_kq}
    When the distortion $D$ satisfies $kq\leq D< \infty$, for all $x\in\mc X_{d,k,q}$, we have \[
    \sum_{r=1}^d \check{x}(r) = \sum_{i=1}^d x(i) \leq kq \leq D.
    \]
    Since $\check{x}(0)=\infty>D$, we have $s^* = 0$ and $i^*=0$. Observe that $s(i)\geq 1$ for all $i\in[d]$, using \eqref{eq::mapping_rule_2}, we have $y(x)(i)=0$ for all $i$ such that $s(i)\geq 1$. That is, $y(x)=\mb 0$ for any $x\in\mc X_{d,k,q}$ when $D\geq kq$, implying that $|\mc Y(D)|=1$ in this case. Therefore, in the rest of this section, it suffices to consider $D<kq$.
\end{remark}

We now prove in the following lemma some useful properties related to each $x\in\mc X_{d,k,q}$, its corresponding $\check{x}$, and $y(x)$ under a distortion constraint $D$. The readers may verify these properties using Example~\ref{example::mapping}.


\begin{lemma}
\label{lemma::covering_mapping}
    Under a distortion constraint $D$, for any $x\in\mc X_{d,k,q}\backslash\{\mb 0\}$ and its corresponding $y(x)$ obtained through \eqref{eq::mapping_order}, \eqref{eq::mapping_rule} and \eqref{eq::mapping_rule_2}, if $y(x)\neq \mb 0$, the following hold: \begin{enumerate}[label=\roman*)]
        \item $m(y(x))=y(x)(i^*)>0$
        \item $\max\{x(i):i\in\supp(x)\backslash\supp(y(x))\}\leq y(x)(i^*)$ 
        \item $\{i\in[d]:x(i)=m(y(x)), y(x)(i)=0\}\subseteq\mc A(y(x))$ 
        \item $\|x-y(x)\|_1\leq D$.
    \end{enumerate}
\end{lemma}

\begin{proof}
First of all, $x\in\mc X_{d,k,q}\backslash\{\mb 0\}$ implies that $\supp(x)\neq \phi$. Throughout this proof, we assume that $y(x)\neq \mb 0$, so that $\supp(y(x))\neq \phi$.

For any $i\in\supp(y(x))$, since $y(x)(i)>0$, by \eqref{eq::mapping_rule_2}, this index $i$ satisfies $s(i)\leq s(i^*)$ and $y(x)(i)=\check{x}(s(i))=x(i)$. In particular, we have \begin{equation*}
    y(x)(i^*)=x(i^*).
\end{equation*} Since $s(i)\leq s(i^*)$, with \eqref{eq::mapping_order}, we obtain that $x(i)\geq x(i^*)=y(x)(i^*)$. Therefore, for all $i\in\supp(y(x))$, we have \[
y(x)(i)=x(i)\geq y(x)(i^*).
\]
Since $y(x)(i^*)=\check{x}(s(i^*))=\check{x}(s^*)>0$, we conclude that the minimum non-zero entry in $y(x)$ has the value \begin{equation}
\label{eq::yx_min_value}
m(y(x)) = y(x)(i^*)>0.
\end{equation}
Hence, we have proved i).

An index $i\geq 1$ satisfying $s(i)\leq s(i^*)$ implies that \[
y(x)(i)\utag{a}{=}x(i)\utag{b}{\geq} x(i^*)=\check{x}(s(i^*))=\check{x}(s^*)>0,
\]
where (a) holds by \eqref{eq::mapping_rule_2} and (b) holds by \eqref{eq::mapping_order}.
That is, $i\in\supp(y(x))$ if $s(i)\leq s(i^*)$. Moreover, the converse is also true since if $s(i)>s(i^*)$, then $y(x)(i)=0$ by \eqref{eq::mapping_rule_2}, implying that $i\notin\supp(y(x))$. Hence, we conclude that, for any index $i\in[d]$, \[
i\in\supp(y(x))\Leftrightarrow s(i)\leq s(i^*),
\] or equivalently, \begin{equation}
\label{eq::equvlt_cond}
i\notin\supp(y(x))\Leftrightarrow s(i)>s(i^*).
\end{equation}
Consider an index $i\in\supp(x)\backslash\supp(y(x))$. Since $i\in\supp(x)\neq \phi$, we have $i\geq 1$. On the other hand, since $i\notin\supp(y(x))$, by \eqref{eq::equvlt_cond}, we have $s(i)>s(i^*)$, which by \eqref{eq::mapping_order} holds if and only if $x(i)<x(i^*)$ or $x(i)=x(i^*)$ and $i>i^*$. Since $y(x)(i^*)=x(i^*)$, we conclude that the maximum $x(i)$ for $i\in\supp(x)\backslash\supp(y(x))$ is no larger than $y(x)(i^*)$, proving ii).

Before proving iii) by contradiction, we first observe that, for any index $i\neq i^*$ with $y(x)(i)=m(y(x))=y(x)(i^*)$, we have $s(i)<s(i^*)$ by \eqref{eq::mapping_rule_2} and $i<i^*$ by \eqref{eq::mapping_order}, implying that $i_m(y(x))=i^*$. Now if the third result is false, then there exists an index $i^\#$ such that $x(i^\#)=m(y(x))$, $y(x)(i^\#)=0$ and $i^\#\notin\mc A(y(x))$, \ie, $i^\#< i_m(y(x))$.
Therefore, $x(i^\#)=m(y(x))=y(x)(i^*)>0$ (cf.~\eqref{eq::yx_min_value}) and $i^\#<i_m(y(x))=i^*$. By \eqref{eq::mapping_order}, we have $s(i^\#)<s(i^*)$ with \[
    y(x)(i^\#)=0\neq x(i^\#), 
    \]
which contradicts \eqref{eq::mapping_rule_2},
implying that $\{i\in[d]:x(i)=m(y(x)), y(i)=0\}\subseteq\mc A(y(x))$, i.e., iii).

Finally, for any $x\in\mc X_{d,k,q}$ and its corresponding $y(x)$ for a given $D$, we have \begin{align*}
\|x-y(x)\|_1 = \sum_{i:s(i)>s(i^*)} x(i)=\sum_{r=s^*+1}^{d}\check{x}(r) \leq D,
\end{align*}
where the first equality holds by \eqref{eq::mapping_rule_2} and the last inequality holds by the definition of $s^*$. This proves iv).
\end{proof}

For this covering compression scheme \[
\pi_{\it cover} := (\mc E_{\it cover}, \mc D_{\it cover}),
\]
we have the deterministic encoder \begin{equation}
\label{eq::cover_encode}
    \mc E_{\it cover}(X_i) = Y(X_i)\in\mc Y(D),
\end{equation}
where the subset $\mc Y(D)$ is the collection of all possible $ y(x)$'s, i.e.,
\begin{equation}
\label{eq::yd_def}
    \mc Y(D) = \big\{ y': y'=  y( x) \text{ for some } x\in\mc X_{d,k,q}\big\}.
\end{equation}
Then each client $i$ transmits $ Y( X_i)$ in $C(\pi_{\it cover})/n=\lceil \log|\mc Y(D)|\rceil$ bits to the server, which can then compute the aggregation 
\begin{equation}
\label{eq::cover_decode}
\hat{ U}_{\it cover}  = \mc D_{\it cover}(\mc E_{\it cover}(X^n)) := \frac{1}{n} \sum_{i=1}^n  \mc E_{\it cover}(X_i).
\end{equation}

We can easily verify that $\pi_{\it cover}\in (D,0)$ by noting that 
\begin{align*}
        \| U-\hat{ U}_{\it cover}\|_1 = \left\|\frac{1}{n}\sum_{i=1}^n \mc E_{\it cover}(X_i) - \frac{1}{n}\sum_{i=1}^n X_i\right\|_1 = \frac{1}{n}\left\|\sum_{i=1}^n (Y( X_i)-X_i)\right\|_1 \leq \max_{ x\in\mc X_{d,k,q}}\|y( x)-x\|_1 \leq D.
    \end{align*}

Now we consider the cardinality $|\mc Y(D)|$ to evaluate the communication cost of this proposed covering scheme.
First of all, if $D\in[0,1)$, since $\check{x}(j)\geq 1>D$ for all $j=1,2,\ldots,\|x\|_0$, $s^*=d$ (cf.~\eqref{eq::mapping_rule}) in this case and $y(x)=x$ for all $x\in\mc X_{d,k,q}$. Hence, $|\mc Y(D)|=|\mc X_{d,k,q}|$ for $D\in[0,1)$. 

We then present a formula of $|\mc Y(D)|$ for $D\geq 1$ obtained through partitioning $\mc Y(D)$.
Define \begin{align*}
    \mc Y_{t,u,a} = \{y\in[0:q]^d: \ |\supp(y)|=t, m(y)=u, |\mc A(y)| = a\},
\end{align*}
\ie, for each $y\in\mc Y_{t,u,a}$, $y$ has the support size equal to $t$, the minimum value of all non-zero entries equal to $u$, and $a$ zero entries after the last entry with value $u$. For the vector \[
y(x)=(0,0,2,0,0,7,0,0,0)
\] in Example~\ref{example::mapping}, since $|\supp(y(x))|=2$, $m(y(x))=2$, and $|\mc A(y(x))|=5$, we see that $y(x)\in\mc Y_{2,2,5}$. 
In the following theorem, we will first prove that the union of $\mc Y_{t,u,a}$'s over certain ranges of $t$, $u$, and $a$ (which depend on $D$) forms a partition of $\mc Y(D)\backslash\{\mb 0\}$. Then the cardinality of $\mc Y(D)$ can be obtained by counting $|\mc Y_{t,u,a}|$ for all possible combinations of $t$, $u$, and $a$.

\begin{theorem}
    \label{thm::yd_size} 
    For any given distortion $D\geq 1$,
    \begin{align}
        |\mc Y(D)| =1+ \sum_{t=1}^{\ceil{k-D/q}}\sum_{\substack{u=\floor{\frac{D}{k-t+1}}+1}}^q \ \sum_{\substack{a=(1+\floor{D-u-(u-1)(k-t)}\})^+}}^{d-t} |\mc Y_{t,u,a}|,\label{eq::yd_size}
    \end{align} 
    where \begin{align*}
    |\mc Y_{t,u,a}| =\sum_{i=d-a-t+1}^{d-a}\binom{d-i}{d-i-a}(q-u)^{d-i-a}\times \binom{i-1}{t-1-(d-i-a)}(q-u+1)^{t-1-(d-i-a)}.        
    \end{align*}
\end{theorem}

\begin{proof}
    Define \begin{align}
    \label{eq::tilde_yd}
        \Tilde{\mc Y}(D) = \bigcup_{t=1}^{\ceil{k-D/q}}\bigcup_{u=\floor{\frac{D}{k-t+1}}+1}^{q} \ \bigcup_{a=(1+\floor{D-u-(u-1)(k-t)})^+}^{d-t}\mc Y_{t,u,a},
    \end{align}
    and we aim to prove that $\mc Y(D)\backslash\{\mb 0\}=\Tilde{\mc Y}(D)$. Evidently, $\mc Y_{t,u,a}\cap\mc Y_{t',u',a'}(D)=\phi$ for $(t,u,a)\neq (t',u',a')$ by definition. Therefore, $\{\mc Y_{t,u,a}\}$ with $(t,u,a)$'s specified in \eqref{eq::tilde_yd} forms a partition of $\Tilde{\mc Y}(D)$.

    We first prove $(\mc Y(D)\backslash\{\mb 0\})\subseteq\Tilde{\mc Y}(D)$. Consider any $y'\in\mc Y(D)$. Then $y'=y(x)$ for some $x\in\mc X_{d,k,q}$. Suppose $|\supp(y(x))|=t$, $m(y(x))=u$, and $|\mc A(y(x))|=a$. Define \[
    \mc T=\supp(x)\backslash\supp(y(x))
    \] 
    with \begin{equation}
    \label{eq::T_size}
        |\mc T|=|\supp(x)|-|\supp(y(x))|\leq k-t
    \end{equation} since $\supp(y(x))\subseteq\supp(x)$. Also, define the index set \begin{align*}
        \mc T_u=\{i:x(i)=u\text{ and }y(x)(i)=0\}.
    \end{align*}
    We have shown in Lemma~\ref{lemma::covering_mapping} that $\mc T_u\subseteq\mc A(y(x))$ and so $|\mc T_u|\leq |\mc A(y(x))|=a$. Since $\mc T_u\subseteq\mc T$, we also have $|\mc T_u|\leq |\mc T|\leq k-t$.
    Therefore, \begin{equation}
        \label{eq::Tu_size}
        |\mc T_u|\leq \min\{a, k-t\}.
    \end{equation} We then have \begin{align*}
        \|x-y(x)\|_1 &=\sum_{i\in\mc T}|x(i)-y(x)(i)|\\
        &=\sum_{i\in\mc T_u}|x(i)-y(x)(i)| + \sum_{i\in\mc T\backslash\mc T_u} |x(i)-y(x)(i)|\\
        &=|\mc T_u|u + \sum_{i\in\mc T\backslash\mc T_u} |x(i)-y(x)(i)|\\
        &\leq |\mc T_u|u + (|\mc T|-|\mc T_u|)(u-1)\\
        &=|\mc T|(u-1) + |\mc T_u|,
    \end{align*}
    where the first inequality above follows from ii) of Lemma~\ref{lemma::covering_mapping}.
    By \eqref{eq::mapping_rule}, $y(x)$ satisfies
    \begin{align}
        D<\sum_{s=s^*}^{d}\check{x}(s) &=\check{x}(s(i^*))+ \sum_{s=s^*+1}^{d}\check{x}(s)\notag \\
        &= x(i^*) + \|x-y(x)\|_1\notag\\
        &=y(x)(i^*) + \|x-y(x)\|_1\notag\\
        &\leq y(x)(i^*) + |\mc T|(u-1) + |\mc T_u|\notag\\
        &= u+|\mc T|(u-1)+|\mc T_u| \label{eq::thm1_temp1}\\
        &\leq u+(k-t)(u-1) + \min\{a, k-t\}\label{eq::range_tua},
    \end{align}
    where \eqref{eq::thm1_temp1} holds by $y(x)(i^*)=m(y(x))=u$ from Lemma~\ref{lemma::covering_mapping}, and the last inequality holds by \eqref{eq::T_size} and \eqref{eq::Tu_size}.

    From \eqref{eq::range_tua} and $\min\{a, k-t\}\leq a$, we obtain \begin{equation}
    \label{eq::range_a_1}
        a>D-u-(k-t)(u-1).
    \end{equation}  
    On the other hand, from \eqref{eq::range_tua} and $\min\{a, k-t\}\leq k-t$, we obtain \begin{equation}
    \label{eq::range_u_1}
        D<u+(u-1)(k-t) + (k-t)=(k-t+1)u.
    \end{equation}
    Also, since $u=m(y(x))\leq \max_i y(x)(i)\leq \max_i x(i)\leq q$, we have \begin{equation}
        \label{eq::range_t_1}
        D<(k-t+1)q
    \end{equation} from \eqref{eq::range_u_1}. 

    By \eqref{eq::range_t_1}, $t$ is the largest integer that is strictly smaller than $k-D/q+1$, i.e., $t$ satisfies \begin{equation}
    \label{eq::range_t}
    1\leq t\leq \left\lceil k-\frac{D}{q}+1\right\rceil-1 = \left\lceil k-\frac{D}{q}\right \rceil.
    \end{equation} 
    Here $k-D/q>0$ since it suffices to consider $D<kq$ as in Remark~\ref{remark::D_le_kq}.

    Now fix $t$ with $t$ satisfying \eqref{eq::range_t}.
    By \eqref{eq::range_u_1}, $u$ should be no less than the smallest integer that is strictly larger than $D/(k-t+1)$, \ie, $u$ satisfies \begin{equation}
    \label{eq::range_u}
    \left\lfloor\frac{D}{k-t+1}\right\rfloor + 1\leq u\leq q.
    \end{equation}
    Here $D<q(k-t+1)$ since $t<k-D/q+1$ by \eqref{eq::range_t}.

    Now fix $t$ and $u$ with $t$ and $u$ satisfying \eqref{eq::range_t} and \eqref{eq::range_u}, respectively.
    By \eqref{eq::range_a_1}, $a$ should be no less than the smallest nonnegative integer that is strictly larger than $D-u-(k-t)(u-1)$, i.e.,  $a \geq (\floor{D-u-(k-t)(u-1)}+1)^+.$
    Also, $a\leq d-t$ since there are $d-t$ zero entries in $y(x)$, implying that $a$ satisfies 
    \begin{equation}
    \label{eq::range_a}
        (\floor{D-u-(k-t)(u-1)}+1)^+ \leq a\leq d-t.
    \end{equation}
    Hence, we have shown that the ranges of $t$, $u$, and $a$ as specified in \eqref{eq::range_t}, \eqref{eq::range_u}, and \eqref{eq::range_a} are necessary for $y'=y(x)\in\mc Y(D)$ with some $x\in\mc X_{d,k,q}$, and as a result, $\mc Y(D)$ is a subset of $\Tilde{\mc Y}(D)$.

    Next, we prove that $\mc Y(D)\backslash\{\mb 0\}\supseteq\Tilde{\mc Y}(D)$ by showing that for every $y'\in\mc Y_{t,u,a}$ with $t$, $u$, and $a$ satisfying the constraints in \eqref{eq::range_t}, \eqref{eq::range_u}, and \eqref{eq::range_a}, $y'$ is in $\mc Y(D)$ defined in \eqref{eq::yd_def}, \ie, there exists an $x\in\mc X_{d,k,q}$ such that $y(x)=y'$.

    Fix a distortion $D$. 
    Recall from \eqref{eq::range_tua} that any tuple $(t,u,a)$ where $t=|\supp(y(x))|$, $u=m(y(x))$, and $a=|\mc A(y(x))|$ for some $x\in\mc X_{d,k,q}$ satisfies \[
    D<u+(k-t)(u-1)+\min\{a,k-t\}.
    \] 
    Under this constraint, we consider three mutually exclusive and exhaustive cases:
    \vspace{1em} 
    
    \noindent \textbf{Case 1} $1\leq D<u$
    
    For this case, we only have to consider $u\geq 2$, because otherwise the case does not exist.
    First of all, $y'\in\mc X_{d,k,q}$ because it follows from \eqref{eq::range_t} that $1\leq t\leq k$.  Let $x=y'$, so that $x\in\mc X_{d,k,q}$. 
    For $i\geq 1$, since $x(i)>0$ if and only if $i\in\supp(x)$, in the permutation $s$ defined in \eqref{eq::mapping_order}, we have \[
    s(i)<s(j) \text{ for any }i\in\supp(x) \text{ and }j\notin\supp(x),
    \]
    implying that $s(i)\leq |\supp(x)|=|\supp(y')|=t$ if and only if $i\in\supp(x)$. In the corresponding $\check{x}$, $\check{x}(t)=m(x)=m(y')=u$ and $\check{x}(r)=0$ for all $r=t+1,\ldots,d$. Since 
    \[
    \sum_{r=t}^d\check{x}(r)=\check{x}(t)=u>D
    \] and \[
    \sum_{r=t+1}^d\check{x}(r)=0,
    \]
    we have $s^*=t$. By \eqref{eq::mapping_rule_2}, $y(x)(i)=x(i)$ for $i$ such that $s(i)\leq s^*=t$ and $y(x)(i)=0$ for $i$ such that $s(i)>t$. Recall that $s(i)\leq t$ if and only if $i\in\supp(x)$. Therefore, $y(x)(i)=x(i)$ for $i\in\supp(x)$ and $y(x)(i)=0$ for $i\notin\supp(x)$, implying that $y(x)=x=y'$. Since $x\in\mc X_{d,k,q}$, this implies that $y'\in\mc Y(D)$ by the definition of $\mc Y(D)$.
    \vspace{1em}

    \noindent \textbf{Case 2} $u\leq D<u+(k-t)(u-1)$
    
    For this case, we only have to consider $t\leq k-1$ and $u\geq 2$, because otherwise the case does not exist.
    For each $t$ and $u$ satisfying the constraint for this case, we always have \[
    u+\alpha(u-1)+u'\leq D<u+\alpha(u-1)+u'+1
    \] for some $\alpha=0,1,\ldots,k-t-1$ and $u'=0,1,\ldots,u-2$.
    Then arbitrarily select an index set $\mc I_{u-1}\subset[d]\backslash\supp(y')$ with $|\mc I_{u-1}|=\alpha$ and an index $i_0\in[d]\backslash(\supp(y')\cup\mc I_{u-1})$. Define $x$ by \begin{equation}
        \label{eq::define_x_case2}
        x(i) = \begin{cases}
            y'(i) &\text{ if } i\in\supp(y')\\
            u-1 &\text{ if } i\in\mc I_{u-1}\\
            u'+1 &\text{ if } i = i_0\\
            0 &\text{ otherwise}
        \end{cases}
    \end{equation}
    for all $i\in[d]$. Note that $u\geq 2$ implies that $u-1\geq 1$ and $u'\leq u-2$ implies that $u'+1\leq u-1$. Since $|\supp(x)|=|\supp(y')|+|\mc I_{u-1}|+1=t+\alpha+1\leq t+(k-t-1)+1 = k$ and $\max_i x(i)= \max_i y'(i)\leq q$, we see that $x\in\mc X_{d,k,q}$.
    
    Next, we prove that $y(x)=y'$. First map this $x$ to $\check{x}$ by the permutation $s$ defined in \eqref{eq::mapping_order}. 
    By our construction of $x$, \[
    x(i)=y'(i)\geq m(y')= u
    \]
    for $i\in\supp(y')$, and \[
    1\leq x(i)\leq u-1
    \]
    for $i\in\mc I_{u-1}\cup\{i_0\}$. 
    Therefore, for any $i\in\supp(y')$ and $j\notin \supp(y')$, we have $x(i)\geq u>x(j)$. Using the permutation $s$ defined by \eqref{eq::mapping_order}, $s(i)<s(j)$ for any $i\in\supp(y')$ and $j\notin \supp(y')$, implying that \begin{equation}
    \label{eq::si_t_iff}
    s(i)\leq |\supp(y')|=t \text{ if and only if } i\in\supp(y').
    \end{equation}
    Moreover, since $i_m(y')\in\supp(y')$, from \eqref{eq::define_x_case2}, we have \begin{equation}
    \label{eq::x_imy_u}
    x(i_m(y'))=y'(i_m(y'))=m(y')=u.
    \end{equation}
    Then for any $i\in\supp(y')$, either \begin{equation}
    \label{eq::temp1_case2}
        x(i_m(y'))=u<y'(i)=x(i),
    \end{equation} or \begin{equation}
    \label{eq::temp2_case2}
        x(i_m(y'))=u=y'(i)=x(i) \text{ with }i\leq i_m(y').
    \end{equation}
    If \eqref{eq::temp1_case2} holds, then $s(i)<s(i_m(y'))$ by \eqref{eq::mapping_order}; if \eqref{eq::temp2_case2} holds, then $s(i)\leq s(i_m(y'))$ by \eqref{eq::mapping_order}. That is, $s(i)\leq s(i_m(y'))$ for any $i\in\supp(y')$.
    Hence, $s(i_m(y'))=|\supp(y')|=t$,
    so that \begin{equation}
    \label{eq::x_check_t_u}
    \check{x}(t)=\check{x}(s(i_m(y')))=x(i_m(y'))=u.
    \end{equation}

    It follows from \eqref{eq::si_t_iff} that 
    \begin{align*}
    \sum_{r=t}^{d} \check{x}(r) &= \check{x}(t) + \sum_{r:r=s(i),i\notin\supp(y')}\check{x}(r)\\
    &=\check{x}(t) + \sum_{i\in\mc I_{u-1}\cup\{i_0\}}\check{x}(s(i))+0\\
    &=\check{x}(t) + \sum_{i\in\mc I_{u-1}}x(i) + x(i_0)\\
    &=u + \alpha(u-1) + u'+1>D
    \end{align*} and
    \begin{align*}
    \sum_{r=t+1}^{d} \check{x}(r) = \sum_{i\in\mc I_{u-1}}x(i) + x(i_0)= \alpha(u-1) + u'+1 < \alpha(u-1) + u'+u\leq D,
    \end{align*}
    where in the last step we have invoked $u\geq 2$.
    Thus $s^*=t$. 
    By \eqref{eq::mapping_rule_2} and \eqref{eq::define_x_case2}, $y(x)(i)=x(i)=y'(i)$ for $i$ such that $s(i)\leq t$, i.e., $i\in\supp(y')$ (cf.~\eqref{eq::si_t_iff}); otherwise, $y(x)(i)=0$. Therefore, $y(x)=y'$ where $x\in\mc X_{d,k,q}$, implying that $y'\in\mc Y(D)$ by the definition of $\mc Y(D)$.
    \vspace{1em}
    
    \noindent \textbf{Case 3} $D\geq u+(k-t)(u-1)$
    
    Note that by \eqref{eq::range_t} and \eqref{eq::range_u}, we have $k-t\geq 0$ and $u\geq 1$.
    Let \begin{equation*}
    \alpha'=\floor{D-u-(k-t)(u-1)}+1.
    \end{equation*}
    For $(t,u,a)$ satisfying the constraint in this case, the lower bound in \eqref{eq::range_a} becomes\[
    (\lfloor D-u-(k-t)(u-1) \rfloor + 1)^+ = \lfloor D-u-(k-t)(u-1)\rfloor + 1
    \]
    since $D-u-(k-t)(u-1)\geq 0$.
    Thus, $\alpha'\leq a$ since $\alpha'$ is the lower bound in \eqref{eq::range_a}. Also, since $D<(k-t+1)u$ by \eqref{eq::range_u_1}, \begin{align*}
    D-u-(k-t)(u-1) =D-(k-t+1)u + (k-t) <k-t,
    \end{align*}
    implying that $\alpha'\leq k-t$.  
    
    Arbitrarily select $\mc I_u\subseteq\mc A(y')$ with $|\mc I_u|=\alpha'$. This is always possible because we have already shown that $\alpha'\leq a=|\mc A(y')|$. Also,  arbitrarily select $\mc I_{u-1}\subset[d]\backslash(\supp(y')\cup\mc I_u)$ with $|\mc I_{u-1}|=k-t-\alpha'$. This is valid since $k-t-\alpha'<d-t-\alpha'=|[d]\backslash(\supp(y')\cup\mc I_u)|$.
    Now define $x$ by 
    \begin{equation}
        \label{eq::define_x_case3}
        x(i) = \begin{cases}
            y'(i) &\text{ if } i\in\supp(y')\\
            u &\text{ if } i\in\mc I_{u}\\
            u-1 &\text{ if } i\in\mc I_{u-1}\\
            0 &\text{ otherwise}
        \end{cases}
    \end{equation}    
    for all $i\in[d]$.
    Note that when $u=1$, $u-1$ becomes $0$ and so $\mc I_{u-1}\subseteq [d]\backslash\supp(x)$.
    Since $|\supp(x)|\leq|\supp(y)|+|\mc I_u|+|\mc I_{u-1}|=t+\alpha'+(k-t-\alpha')=k$, $\max_i x(i)=\max_i y'(i)\leq q$, and $u-1\geq 0$ due to \eqref{eq::range_u}, we see that $x\in\mc X_{d,k,q}$.

    Next, we prove that $y(x)=y'$.
    While the proof of this claim resembles the corresponding proof for Case 2, the details are different.
    First map this $x$ to $\check{x}$ by the permutation $s$ defined in \eqref{eq::mapping_order}.
    By our construction of $x$, \[
    x(i)\geq u \text{ for } i\in\supp(y')\cup\mc I_u
    \]
    and
    \[x(i)\leq u-1, \text{ otherwise.}\] Hence, for $i\in\supp(y')$ and $j\notin\supp(y')\cup\mc I_u$, we have \[
    x(i)\geq u>x(j),
    \]
    implying that $s(i)<s(j)$ by \eqref{eq::mapping_order}.
    For $i\in\supp(y')$ with $x(i)=y(i)>u$ and $j\in\mc I_u$, we have \[
    x(i)>u=x(j),
    \]
    implying $s(i)<s(j)$ as well by \eqref{eq::mapping_order}.    
    Moreover, for $i\in\supp(y')$ with $x(i)=y'(i)=u$ and $j\in\mc I_u$, since $\mc I_u\subseteq\mc A(y')$, we have \[
    x(i)=u=x(j) \text{ with } i\leq i_m(y')<j,
    \]
    implying that $s(i)<s(j)$ still holds by \eqref{eq::mapping_order}.
    Therefore, $s(i)<s(j)$ for any $i\in\supp(y')$ and $j\notin\supp(y')$, and as a result, \begin{equation}
    \label{eq::sit_iff_case3}
    s(i)\leq |\supp(y')| = t \text{ if and only if }i\in\supp(y').
    \end{equation}
    Moreover, since $i_m(y')\in\supp(y')$, from \eqref{eq::define_x_case3}, we have \[
    x(i_m(y'))=y'(i_m(y'))=m(y')=u,
    \]
    i.e., \eqref{eq::x_imy_u}. Then using exactly the same argument following \eqref{eq::x_imy_u}, we can obtain \eqref{eq::x_check_t_u}, i.e., $\check{x}(t)=u$.

    We then have, \begin{align*}
    \sum_{r=t}^d \check{x}(r) &= \check{x}(t) + \sum_{r: r=s(i), i\notin\supp(y')} \check{x}(r)\\
    &= \check{x}(t) + \sum_{i\in\mc I_u\cup\mc I_{u-1}} \check{x}(s(i)) + 0 \\
    &= \check{x}(t) + \sum_{i\in\mc I_u} x(i) + \sum_{i\in\mc I_{u-1}}x(i)\\
    &= u + \alpha' u+ (k-t-\alpha')(u-1)\\
    &=u+(k-t)(u-1) + \alpha'\\
    &=u+(k-t)(u-1) + 1+ \lfloor D-u-(k-t)(u-1)\rfloor \\
    &>u+(k-t)(u-1) + 1+ D - u - (k-t)(u-1) -1\\
    & = D
    \end{align*} and    
    \begin{align*}
    \sum_{r=t+1}^d \check{x}(r) &= \sum_{i\in\mc I_u} x(i) + \sum_{i\in\mc I_{u-1}} x(i)\\
    &= \alpha' u+ (k-t-\alpha')(u-1)\\
    &= (k-t)(u-1) + \alpha'\\
    &= (k-t)(u-1) + 1+ \lfloor D-u-(k-t)(u-1)\rfloor\\
    &\leq (k-t)(u-1) + 1+ D-u-(k-t)(u-1)\\
    &= 1+D-u \leq D,
    \end{align*}
    implying that $s^*=t$. By \eqref{eq::mapping_rule_2} and \eqref{eq::define_x_case3}, $y(x)(i)=x(i)=y'(i)$ for $i$ such that $s(i)\leq t$, i.e., $i\in\supp(y')$ (cf.~\eqref{eq::sit_iff_case3}); otherwise, $y(x)(i)=0$. Therefore, $y(x)=y'$ where $x\in\mc X_{d,k,q}$, implying that $y'\in\mc Y(D)$ by the definition of $\mc Y(D)$.
    \vspace{1em}
    
    Therefore, for any $y'\in\Tilde{\mc Y}(D)$, there exists a $x\in\mc X_{d,k,q}$ such that $y(x)=y'$, implying that $y'\in\mc Y(D)\backslash\{\mb 0\}$ and $\mc Y(D)\backslash\{\mb 0\}\supseteq\Tilde{\mc Y}(D)$.
    Recalling that $\Tilde{\mc Y}(D)\supseteq\mc Y(D)\backslash\{\mb 0\}$, we conclude that $\mc Y(D)\backslash\{\mb 0\}=\Tilde{\mc Y}(D)$.

    It remains to determine the cardinality $|\mc Y_{t,u,a}|$ for every possible combination of $t$, $u$, and $a$ satisfying \eqref{eq::range_t}, \eqref{eq::range_u}, and \eqref{eq::range_a}, respectively. 
    Consider the index $i_m(y')$ for a $y'\in\mc Y_{t,u,a}$. Let $i_m(y')=i$. Then there are $(d-i)$ entries after the last occurrence of the value $u$, among which there are $a$ zero entries and $(d-i-a)$ non-zero entries. Those non-zero entries after $y'(i)$ take values in $\{u+1,u+2,\ldots,q\}$. On the other hand, there are $(t-1-(d-i-a))$ non-zero entries before $y'(i)$, taking values in $\{u,u+1,\ldots,q\}$. From $d-i-a\geq 0$ and $t-1-(d-i-a)\geq 0$, we have \[
    d-a-t+1 \leq i\leq d-a.
    \] 
    Since there are $\binom{d-i}{d-i-a}(q-u)^{d-i-a}$ ways to choose the non-zero entries after $y'(i)$ and $\binom{i-1}{t-1-(d-i-a)}(q-u+1)^{t-1-(d-i-a)}$ ways to choose the non-zero entries before $y'(i)$, we have
    \begin{align}
    \label{eq::y_tma_size}
        |\mc Y_{t,u,a}| =\sum_{i=d-a-t+1}^{d-a}\binom{d-i}{d-i-a}(q-u)^{d-i-a}\times \binom{i-1}{t-1-(d-i-a)}(q-u+1)^{t-1-(d-i-a)}.
    \end{align} 

    In conclusion, the set $\mc Y(D)$ constructed by our proposed covering scheme has cardinality equal to \begin{align*}
        1+ \sum_{(t,u,a)}|\mc Y_{t,u,a}|= 1+\sum_{t=1}^{\ceil{k-D/q}}\sum_{u=\floor{\frac{D}{k-t+1}}+1}^q\sum_{\substack{a=\\(1+\floor{D-u-(u-1)(k-t)}\})^+}}^{d-t} |\mc Y_{t,u,a}|.
    \end{align*}
    Then Theorem~\ref{thm::yd_size} is proved.
\end{proof}

Since the cardinality of $\mc Y(D)$ in \eqref{eq::yd_size} is difficult to evaluate, we prove an upper bound on $\log |\mc Y(D)|$, which is the communication cost per client of this covering scheme.
\begin{corollary}
\label{cor::cover_size}
    For any distortion $D\geq 0$, 
    \begin{equation*}
        \log |\mc Y(D)|\leq \lceil k-D/q \rceil \log(dq)+1.
    \end{equation*}
\end{corollary}
\begin{proof}
    Consider \begin{align*}
        |\mc Y_{t,u,a}| &= \sum_{i=d-a-t+1}^{d-a}\binom{d-i}{d-i-a}(q-u)^{d-i-a}\times \binom{i-1}{t-1-(d-i-a)}(q-u+1)^{t-1-(d-i-a)}\\
        &\leq \sum_{i=d-a-t+1}^{d-a}\binom{d-i}{d-i-a}\binom{i-1}{t-1-(d-i-a)}(q-u+1)^{t-1}\\
        &= \sum_{j=0}^{t-1}\binom{a+j}{j}\binom{d-1-(a+j)}{t-1-j}(q-u+1)^{t-1},
    \end{align*}
    where the last equality is obtained by letting $j=d-i-a$. Therefore, 
    \begingroup
    \allowdisplaybreaks
    \begin{align*}
        \sum_{\substack{a=(1+\floor{D-u-(u-1)(k-t)}\})^+}}^{d-t} |\mc Y_{t,u,a}|&\leq \sum_{a=0}^{d-t}\sum_{j=0}^{t-1}\binom{a+j}{j}\binom{d-1-(a+j)}{t-1-j}(q-u+1)^{t-1}\\
        &= (q-u+1)^{t-1}\sum_{j=0}^{t-1}\sum_{a=0}^{d-t}\binom{a+j}{j}\binom{d-1-(a+j)}{t-1-j}\\
        &=(q-u+1)^{t-1}\sum_{j=0}^{t-1}\sum_{m=j}^{d-t+j}\binom{m}{j}\binom{d-1-m}{t-1-j}\\
        &\leq (q-u+1)^{t-1}\sum_{j=0}^{t-1}\sum_{m=0}^{d-1}\binom{m}{j}\binom{d-1-m}{t-1-j}\\
        &=(q-u+1)^{t-1}\sum_{j=0}^{t-1}\binom{d}{t}\\
        &=t\binom{d}{t}(q-u+1)^{t-1},
    \end{align*}
    \endgroup
    where the second last equality holds by the Chu–Vandermonde identity. Then from \eqref{eq::yd_size}, we have \begin{align*}
        |\mc Y(D)|&=\sum_{t=1}^{\ceil{k-D/q}}\sum_{u=\floor{\frac{D}{k-t+1}}+1}^q\sum_{\substack{a=\\(1+\floor{D-u-(u-1)(k-t)}\})^+}}^{d-t} |\mc Y_{t,u,a}|\\
        &\leq \sum_{t=1}^{\ceil{k-D/q}}\sum_{u=\floor{\frac{D}{k-t+1}}+1}^q t\binom{d}{t}(q-u+1)^{t-1}\\
        &\leq \sum_{t=1}^{\ceil{k-D/q}}t\binom{d}{t}\sum_{u=\floor{\frac{D}{k-t+1}}+1}^q \left(q-\left\lfloor{\frac{D}{k-t+1}}\right\rfloor\right)^{t-1}\\
        &= \sum_{t=1}^{\ceil{k-D/q}}t\binom{d}{t} \left(q-\left\lfloor{\frac{D}{k-t+1}}\right\rfloor\right)^{t}.
    \end{align*}
    Here $q-\lfloor D/(k-t+1)\rfloor>0$ since $q(k-t+1)>D$ by \eqref{eq::range_t_1} in the proof of Theorem~\ref{thm::yd_size}.

Under the assumption that $d\gg k$ so that \[
\binom{d}{k}\approx \frac{d^k}{k!},
\] 
using Stirling's approximation, we have \begin{align*}
    \log \,|\mc Y(D)| &\leq \log \sum_{t=1}^{\ceil{k-D/q}}t\binom{d}{t} \left(q-\left\lfloor{\frac{D}{k-t+1}}\right\rfloor\right)^{t}\\
    &\approx \log \sum_{t=1}^{\ceil{k-D/q}}\frac{d^t}{(t-1)!} \left(q-\left\lfloor{\frac{D}{k-t+1}}\right\rfloor\right)^{t}\\
    &\leq \log \left(q^{\ceil{k-D/q}}\cdot \sum_{t=1}^{\ceil{k-D/q}}\frac{d^t}{(t-1)!}\right)\\
    &\leq \log \left(q^{\ceil{k-D/q}}\cdot \ceil{k-D/q}\cdot \frac{d^{\ceil{k-D/q}}}{(\ceil{k-D/q}-1)!}\right).
\end{align*}
For $\lceil k-D/q \rceil=1$ and $\lceil k-D/q \rceil\geq 4$, since $(\lceil k-D/q \rceil - 1)!\geq \lceil k-D/q \rceil$, we have \begin{align*}
    \log \,|\mc Y(D)| &\leq \log \left(q^{\ceil{k-D/q}}\cdot \ceil{k-D/q}\cdot \frac{d^{\ceil{k-D/q}}}{(\ceil{k-D/q}-1)!}\right)\\
    &\leq \log \left(q^{\ceil{k-D/q}}\cdot d^{\ceil{k-D/q}}\right)\\
    &=\ceil{k-D/q}\log (dq). 
\end{align*} 
For $\lceil k-D/q \rceil=3$, we have \begin{align*}
    \sum_{t=1}^{\ceil{k-D/q}}\frac{d^t}{(t-1)!} = d + d^2 + \frac{d^3}{2} \leq d^3=d^{\lceil k-D/q\rceil}
\end{align*}
for $d\geq 3$. For $\lceil k-D/q\rceil=2$, we have \begin{align*}
    \sum_{t=1}^{\ceil{k-D/q}}\frac{d^t}{(t-1)!} = d + d^2 \leq 2d^2 = 2d^{\lceil k-D/q\rceil}
\end{align*}
for $d\geq 1$.
Hence, we conclude that \begin{align*}
    \log \,|\mc Y(D)| &\leq \left(\log q^{\ceil{k-D/q}}\cdot 2d^{\ceil{k-D/q}}\right)=\ceil{k-D/q}\log (dq) + 1.
\end{align*}

\end{proof}

One can observe that in the proposed covering compression scheme, the largest number of non-zero entries of a $y\in\mc Y(D)$ is $\lceil k-D/q\rceil$, which decreases approximately linearly in the distortion $D$. In Section~\ref{sec::combination}, this property will be used to compute the communication-accuracy tradeoff of sketching methods. Moreover, the relation between the achievable communication cost $n\lceil\log |\mc Y(D)|\rceil$, and the distortion $D$ is approximately linear as well. We will later compare this achievable communication cost of the covering scheme with the converse derived in Section~\ref{sec::converse}.

Moreover, when $D=0$, $\mc C(\pi_{\it cover})\leq n\lceil \log |\mc Y(0)|\rceil\leq n(k(\log (dq)+1)+1)$. Recall that the communication cost of the scheme using Reed-Solomon code mentioned in Section~\ref{sec::related_works} is upper bounded by $n(2k+1)(\log(dq)+1)$ (by relaxing ceil functions using $\lceil x\rceil\leq x+1$), which is about twice of that of $\mc C(\pi_{\it cover})$ under the same distortion $D=0$. This difference might come from the linearity of the encoder $Ax$ of Reed-Solomon code. The proposed $\mc E_{\it cover}$ is not linear, which may help us to reduce the achievable communication cost. 

Moreover, one can also adjust the lossless compression scheme using Reed-Solomon code to a lossy compression scheme by reducing the original sparsity $k$ according to the allowed distortion $D$: let $k'=\lceil k-D/q\rceil$; when compressing a vector $x\in\mc X_{d,k,q}$, let any $(\|x\|_0-k')_+$ entries of its non-zero entries to be $0$ and then use the Reed-Solomon code lossless compression scheme to transmit this new sparser vector. Therefore, the distortion in this case is at most \[
(\|x\|_0-k')_+\cdot q\leq \left(k-\left\lceil k-\frac{D}{q}\right\rceil\right)q\leq D.
\] 
Since each input vector of the Reed-Solomon code now has at most $k'$ non-zero entries, the communication cost per client becomes \begin{align*}
(2k'+1)\left\lceil \log q^{\lceil \log_q d\rceil}\right\rceil &\leq \left(2\left\lceil k-\frac{D}{q}\right\rceil+1\right)\cdot \lceil(\log_qd+1)\cdot \log q\rceil\\
&\leq \left(2\left\lceil k-\frac{D}{q}\right\rceil+1\right)\cdot \lceil\log d + \log q\rceil\\
&\leq \left(2\left\lceil k-\frac{D}{q}\right\rceil+1\right)\cdot\left(\log (dq)+1\right).
\end{align*}
We use Fig.~\ref{fig::delta_0_compare} to compare numerically the achievable communication-accuracy tradeoff between our covering scheme, Reed-Solomon code scheme, and the lower bounds when $\delta=0$, which will be derived in Section~\ref{sec::converse}.


When taking privacy and security requirements into consideration, however, this covering scheme is no longer sufficient since the server can always distinguish each client's local model with confidence. To satisfy the privacy requirement, it is necessary to involve randomness in the compression scheme and one popular strategy is sketching, especially the Count Sketch (CS) and the Count-Min Sketch (CMS). In the rest of this section, we will introduce randomness into our proposed covering scheme by combining it with modified CS and CMS. Before analyzing the communication-accuracy tradeoff of the combined compression scheme, the next section will describe the two modified sketching methods and prove their achievable tradeoff. 

\subsection{Two Sketching Schemes}
\label{sec::two_sketchings_schemes}
The first analysis of CS and CMS on the relation between estimation accuracy and number of counters (closely related to communication cost in this paper) can be found in \cite{charikar2002finding} and \cite{cormode2005improved}. As discussed in Section~\ref{sec::related_works}, the recent work \cite{chen2022breaking} applies CS to the binary $1$-sparse aggregation problem and modifies the original CS with an additional threshold step to further reduce the communication cost. 

We will first refine and generalize the results in \cite{chen2022breaking} by proving the achievable communication-accuracy tradeoff of CS with threshold in Lemma~\ref{lemma::count_sketch}. Then, we will apply the same threshold modification to CMS and analyze its communication-accuracy tradeoff in Lemma~\ref{lemma::cm}. 


\vspace{1em}

\noindent\textbf{Count Sketch:}
The communication scheme using the CS method proposed in \cite{charikar2002finding} is as follows.
Suppose all participants in the distributed system share hash functions $h^t:[d]\rightarrow[w]$, where $w<d$, and $\sigma^t:[d]\rightarrow\{-1,1\}$, $t\in[T]$, which are chosen uniformly at random.
Note that $(\sigma^t(j))^2=1$, whether $\sigma^t(j)=-1$ or $1$.
All hash functions $h^t$ and $\sigma^t$ are mutually independent.
In the $t$-th round, the client $i$ first compresses its $d$-dimensional local model $X_i$ to $Y_i^t$, which is a $w$-dimensional vector given by \begin{equation}
\label{eq::cs_encode}
    Y_i^t(r)= \sum_{j=1}^d \mathds{1}(h^t(j)=r)\sigma^t(j)X_i(j)
\end{equation} for every coordinate $r\in[w]$.  

Run this process in parallel for $T$ rounds with hash functions $h^1,\ldots,h^T$ and $\sigma^1,\ldots,\sigma^T$. Then the server  reconstructs each local model $X_i$ by \begin{equation*}
    \hat{X}_i(j) = \text{median}(\hat{X}_i^1(j),\ldots,\hat{X}_i^T(j))
\end{equation*} for every coordinate $j\in[d]$, where \[
\hat{X}_i^t(j) = \sigma^t(j)Y_i^t(h^t(j)).
\]
Then server reconstructs the global model $U$ based on $(\hat{X}_1,\ldots,\hat{X}_n)$. For this compression using CS, we denote it by $\pi_{\it CS}=(\mc E_{\it CS}, \mc D_{\it CS})$, where $\mc E_{\it CS}(X_i) := (Y_i^t, t\in [T])$ with $Y_i^t$ defined by \eqref{eq::cs_encode}. The decoder $\mc D_{\it CS}$ will be specified in the following proof.

\begin{lemma}
\label{lemma::count_sketch}
    With the CS described above, for any $\delta>0$ and $D>0$, as long as \begin{equation}
    \label{eq::dim_err_distortion}
        wT \geq \left\lceil\frac{200k^3q^2}{3D^2}\right\rceil\cdot\left\lceil\log\left(\frac{d}{\delta}\right)\right\rceil,
    \end{equation} we can reconstruct the global model $U$ by an estimate $\hat{U}_{\it CS}=\mc D_{\it CS}(\mc E_{\it CS}(X^n))$ with \[
        \pr(\|\hat{U}_{\it CS}-U\|_1\leq D)\geq 1-\delta.
    \]
\end{lemma}
\begin{proof}
    In round $t$, the server reconstructs the $j$-th coordinate of $X_i$, $i\in[n]$, by \begin{align*}
       \hat{X_i}^t(j) &= \sigma^t(j)Y_i^t(h^t(j))\\
       &= \sigma^t(j)\left(\sum_{s=1}^d \mathds{1}(h^t(s)=h^t(j))\sigma^t(s)X_i(s)\right)\\
       &=\mathds{1}(h^t(j)=h^t(j))\sigma^t(j)\sigma^t(j)X_i(j) + \sum_{s\neq j}\mathds{1}(h^t(s)=h^t(j))\sigma^t(j)\sigma^t(s)X_i(s)\\
       &= X_i(j) + \sum_{s\neq j}\mathds{1}(h^t(s)=h^t(j))\sigma^t(j)\sigma^t(s)X_i(s).
    \end{align*} 
    Since $\bb E[\mathds{1}(h^t(s)=h^t(j))]=1/w$ and $\bb E[\sigma^t(s)]=0$ for any $s\neq j$, we have \begin{align*}
        \bb E[\hat{X}_i^t(j)]&= X_i(j) + \sum_{s\neq j}\bb E[\mathds{1}(h^t(s)=h^t(j))\sigma^t(j)\sigma^t(s)X_i(s)] \\
        &= X_i(j) + \sum_{s\neq j}\bb E[\mathds{1}(h^t(s)=h^t(j))]\bb E[\sigma^t(j)]\bb E[\sigma^t(s)]X_i(s)\\
        &= X_i(j) + \sum_{s\neq j}\frac{1}{w}\cdot 0\cdot 0\cdot X_i(s) = X_i(j),
    \end{align*} and \begin{align*}
        \text{Var}(\hat{X}_i^t(j)) &= \bb E[(\hat{X}_i^t(j))^2] - \bb E[\hat{X}_i^t(j)]^2\\
        &= \bb E[(\hat{X}_i^t(j))^2] - X_i(j)^2\\
        &= \bb E\left[\left(X_i(j)+\sum_{s\neq j}\mathds{1}(h^t(s)=h^t(j))\sigma^t(j)\sigma^t(s)X_i(s)\right)^2\right] - X_i(j)^2\\
        &\utag{i}{=} \bb E[X_i(j)^2]+ \bb E\left[\left(\sum_{s\neq j}\mathds{1}(h^t(s)=h^t(j))\sigma^t(j)\sigma^t(s)X_i(s)\right)^2\right] - X_i(j)^2\\
        &= \bb E\left[\left(\sum_{s\neq j}\mathds{1}(h^t(s)=h^t(j))\sigma^t(j)\sigma^t(s)X_i(s)\right)^2\right]\\
        &\utag{ii}{=} \bb E\left[\sum_{s\neq j}\mathds{1}(h^t(s)=h^t(j))^2\sigma^t(j)^2\sigma^t(s)^2X_i(s)^2\right]\\
    &= \bb E\left[\sum_{s\neq j}\mathds{1}(h^t(s)=h^t(j))X_i(s)^2\right]\\
    &=\sum_{s\neq j}\bb E[\mathds{1}(h^t(s)=h^t(j))]X_i(s)^2\\
    &= \frac{\sum_{s\neq j}X_i(s)^2}{w}\leq \frac{\|X_i\|_2^2}{w}\leq \frac{kq^2}{w},
    \end{align*}
    where (i) and (ii) hold since $\bb E[\sigma^t(s)]=0$ for all $s\in[d]$.
    By Chebyshev's inequality, \begin{align*}
    \pr\bigg(|\hat{X}_i^t(j)-X_i(j)|>\frac{D}{2k}\bigg)\leq \frac{\text{Var}(\hat{X}_i^t(j))}{(D/(2k))^2} \leq \frac{kq^2/w}{(D/(2k))^2}= \frac{4q^2k^3}{wD^2}.
\end{align*}
    Let $\hat{X}_i(j)$ be the median of $\hat{X}_i^1(j),\ldots,\hat{X}_i^T(j)$ and \begin{equation}
    \label{eq::def_p}
        p=\frac{4q^2k^3}{wD^2}.
    \end{equation} Note that for fixed $i\in[n]$ and $j\in[d]$, the random variables $\hat{X}_i^t(j), t\in[T]$ are i.i.d.
    Then we have \begin{align}
    \label{eq::chernoff_bd}
        \pr\bigg(|\hat{X}_i(j)-X_i(j)|>\frac{D}{2k}\bigg)\notag &\utag{i}{\leq} \pr\left(\rm{Binom}(T, p)\geq \frac{T}{2}\right)\notag\\
        &\utag{ii}{\leq} \inf_{a>0}\,\left[M(a)\cdot e^{-\frac{T}{2}a}\right]\notag\\
        &= \inf_{a>0} \,\left[(1-p+pe^a)^T\cdot e^{-\frac{T}{2}a}\right]\notag\\
        &= \left(\inf_{a>0}\,\left[(1-p)e^{-\frac{a}{2}}+pe^{\frac{a}{2}}\right]\right)^T\notag\\
        &\utag{iii}{=}\left((1-p)\left(\frac{1-p}{p}\right)^{-\frac{1}{2}}+p\left(\frac{1-p}{p}\right)^{\frac{1}{2}}\right)^T\notag\\
        &= \left(2(p(1-p))^\frac{1}{2}\right)^T =\left(4p(1-p)\right)^\frac{T}{2},
    \end{align}
    where (i) holds since $|\hat{X}_i(j)-X_i(j)|>\frac{D}{2k}$ implies that $X_i^t(j)$ satisfies $|\hat{X}_i^t(j)-X_i(j)|>\frac{D}{2k}$ for at least half of all $t\in[T]$. The inequality (ii) holds by the Chernoff bound, where $M(a)$ is the moment generating function of $\rm{Binom}(T,p)$. The expression in (iii) is obtained by setting the derivative of $(1-p)^{-\frac{a}{2}}+pe^{\frac{a}{2}}$ with respect to $a$ to $0$.
    Letting $w=\lceil\frac{200q^2k^3}{3D^2}\rceil$, so that $p\leq 0.06$ (\cf \eqref{eq::def_p}), we have \[
    p(1-p)<\frac{1}{16},
    \]
    implying that \[
    \pr(|\hat{X}_i(j)-X_i(j)|>D/(2k))< \left(4\times \frac{1}{16}\right)^{-\frac{T}{2}}=2^{-T}
    \]
    by \eqref{eq::chernoff_bd}.
    Taking the union bound over $j\in[d]$, for all $X_i$, we have \begin{equation}
    \label{eq::union_bd}
        \pr\left(\max_{j\in [d]}|\hat{X}_i(j)-X_i(j)|>\frac{D}{2k}\right)\leq d\cdot 2^{-T}.
    \end{equation}
    Consider the event \[
    \mc A := \left\{\max_{j\in[d]}|\hat{X}_i(j)-X_i(j)|\leq \frac{D}{2k}\right\},
\]
i.e., the event that \begin{equation}
\label{eq::event_A}
    |\hat{X}_i(j)-X_i(j)|\leq \frac{D}{2k} \text{ for all } j\in[d].
\end{equation}
We threshold out those coordinates $j$ in $\hat{X}_i$ with $|\hat{X}_i(j)|\leq D/(2k)$ and denote the new estimate as $\Tilde{X}_i(j)$, i.e., let \begin{align}
\label{eq::threshold}
    \begin{cases}
    \Tilde{X}_i(j)=0, &  \text{ if } |\hat{X}_i(j)|\leq \frac{D}{2k}\\
    \Tilde{X}_i(j)=\hat{X}_i(j), &  \text{ if } |\hat{X}_i(j)|>\frac{D}{2k}.
\end{cases}
\end{align}
Given the event $\mc A$, i.e., \eqref{eq::event_A} holds, if $X_i(j)=0$, we have \[
|\hat{X}_i(j)|=|\hat{X}_i(j)-0|=|\hat{X}_i(j)-X_i(j)|\leq \frac{D}{2k},
\]
implying that $\Tilde{X}_i(j)=0$ by \eqref{eq::threshold} and \begin{equation}
\label{eq::case_0}
|\Tilde{X}_i(j)-X_i(j)|=|0-0|=0.
\end{equation}
If $X_i(j)\neq 0$ and $|\hat{X}_i(j)|\leq \frac{D}{2k}$, by \eqref{eq::threshold}, we have $\Tilde{X}_i(j)=0$ and \begin{align}
\label{eq::d/k_case_1}
|\Tilde{X}_i(j)-X_i(j)|&\leq |\Tilde{X}_i(j)-\hat{X}_i(j)| + |\hat{X}_i(j)-X_i(j)|\notag\\
& = |0-\hat{X}_i(j)| + |\hat{X}_i(j)-X_i(j)|\notag\\
& \utag{i}{\leq} |\hat{X}(j)| + \frac{D}{2k} \utag{ii}{\leq} \frac{D}{2k} + \frac{D}{2k} = \frac{D}{k},
\end{align}
where (i) holds by \eqref{eq::event_A} and (ii) holds by the assumption on $X_i(j)$ for this case.
If $X_i(j)\neq 0$ and $|\hat{X}_i(j)|>\frac{D}{2k}$, by \eqref{eq::threshold}, we have $\Tilde{X}_i(j)=\hat{X}_i(j)$, implying that \begin{equation}
\label{eq::d/k_case_2}
|\Tilde{X}_i(j)-X_i(j)| = |\hat{X}_i(j)-X_i(j)|\leq \frac{D}{2k},
\end{equation}
where the inequality holds by \eqref{eq::event_A}.
Thus, the event $\mc A$ implies that\begin{itemize}
    \item if $X_i(j)=0$, then \[
|\Tilde{X}_i(j)-X_i(j)| = 0
\]
by \eqref{eq::case_0};
    \item if $X_i(j)\neq 0$, then \[
|\Tilde{X}_i(j)-X_i(j)| \leq \frac{D}{k}
\]
by \eqref{eq::d/k_case_1} and \eqref{eq::d/k_case_2}.
\end{itemize} 
Therefore, given the event $\mc A$, for any $i\in[n]$, \begin{align*}
\|\Tilde{X}_i-X_i\|_1 = \sum_{j:X_i(j)\neq 0}|\Tilde{X}_i(j)-X_i(j)| + \sum_{j:X_i(j)=0}|\Tilde{X}_i(j)-X_i(j)|\leq k\cdot \frac{D}{k} + 0 = D.
\end{align*} 
Letting $T=\lceil\log \left(\frac{d}{\delta}\right)\rceil$ and \[
    \hat{U}=\frac{1}{n}\sum_{i=1}^n \Tilde{X}_i,
    \] 
we have \begin{align*}
        \pr (\|\hat{U}-U\|_1\leq D) \geq \pr\left(\max_{i\in[n]}\|\Tilde{X}_i-X_i\|_1\leq D\right)\geq \pr(\mc A)\utag{i}{\geq} 1-d\cdot 2^{-T}\geq 1-\delta,
    \end{align*}
where the inequality (i) holds by \eqref{eq::union_bd}.
\end{proof}

Since for each $i\in[n]$, $t\in[T]$, and $r\in[w]$, \[
    Y_i^t(r) = \sum_{j=1}^d \mathds{1}(h^t(j)=r)\sigma^t(j)X_i(j)\leq \sum_{j=1}^d X_i(j) \leq kq,
\]
and \[
    Y_i^t(r) = \sum_{j=1}^d \mathds{1}(h^t(j)=r)\sigma^t(j)X_i(j)\geq \sum_{j=1}^d (-X_i(j)) \geq -kq,
\]
we have $Y_i^t(j)\in\{-kq,\ldots,-1,0,1,\ldots,kq\}$.
Thus, $\mc C_{\it CS}$, the communication cost of this Count Sketch scheme for transmitting $(Y^1_i,\ldots,Y^T_i)_{i\in[n]}$, is \begin{align}
\label{eq::cc_CS}
    n(wT\log(2kq+1)) = n\cdot \left\lceil\frac{200k^3q^2}{3D^2}\right\rceil\cdot\left\lceil\log\left(\frac{d}{\delta}\right)\right\rceil\cdot\left\lceil\log(2kq+1)\right\rceil,
\end{align}
where $n$ is the number of clients, $w$ is the dimension of each vector $Y_i^t$, $t\in[T]$, and $\log (2kq+1)$ is the number of bits required in transmitting the value of each coordinate $Y_i^t(r)$, $r\in[w]$.

We can first observe that the communication cost per client of CS grows in the order $\mc O(
\frac{k^3q^2\log(kq)\log d}{D^2})$. Moreover, as we mentioned at the end of Section~\ref{sec::covering}, since $\mc C_{\it CS}$ grows in the order $\mc O(k^3\log k)$ in the number of non-zero entries $k$, replacing the CS input $x_i\in\mc X_{d,k,q}$ by $y(x_i)\in\mc Y(D')$ for some $D'>0$ provides a tradeoff in the communication cost between the distortions $D$ and $D'$ from the two compression steps, which will be explained in detail in the next section.

Before expounding CMS with threshold, we highlight the differences between CS and CMS here. First, unlike CS, the estimate from CMS is not unbiased; it overestimates $X_i(j)$, the value of each coordinate in the local models (see \eqref{eq::cms_larger_estimate} in the proof of Lemma~\ref{lemma::cm}). Also, it is agreed by previous research that CS is usually more efficient when the distortion is measured in $\ell_2$-loss and CMS is more efficient when $\ell_1$-loss is employed. Since we measure distortion using $\ell_1$-loss throughout this paper, this difference in distortion measurements may be one of the reasons for CMS outperforming CS in terms of the communication-accuracy tradeoff. The performance comparison can be found in Fig.~\ref{fig::results} after we have proved Lemma~\ref{lemma::cm}.

\vspace{1em}
\noindent\textbf{Count-Min Sketch} \cite{cormode2005improved}: 
Applying CMS requires $T$ uniform and mutually independent hash functions $h^t:[d]\rightarrow[w]$, where $w<d$, which are shared by all participants in the system. In the $t$-th round, $t\in[T]$, the client $i$ compresses its local model $X_i$ to $Y^t_i$, which is a $w$-dimensional vector given by \begin{equation}
\label{eq::cms_encode}
Y_i^t(r) = \sum_{j=1}^d \mathds{1}(h^t(j)=r)X_i(j),
\end{equation}
for every coordinate $r\in[w]$.

Similar to CS, we run this process in parallel for $T$ rounds with hash functions $h^1,\ldots,h^T$. After receiving $(Y_i^1,\ldots,Y_i^T)_{i\in[n]}$, the server reconstructs each $X_i$ by \[
\hat{X}_i(j) = \min(\hat{X}_i^1(j),\ldots,\hat{X}_i^T(j))
\] for every coordinate $j\in[d]$, where \[
\hat{X}_i^t(j) = Y^t_i(h^t(j)).
\] 
Then server reconstructs the global model $U$ based on $(\hat{X}_1,\ldots,\hat{X}_n)$. For this compression using CMS, we denote it by $\pi_{\it CMS}=(\mc E_{\it CMS}, \mc D_{\it CMS})$, where $\mc E_{\it CMS}(X_i) := (Y_i^t, t\in [T])$ with $Y_i^t$ defined by \eqref{eq::cms_encode}. The decoder $\mc D_{\it CMS}$ will be specified in the following proof.

\begin{lemma}
\label{lemma::cm}
    With the CMS described above, for any $\delta>0$ and $D>0$, as long as \begin{align*}
        wT\geq \left\lceil\frac{4k^2q}{D}\right\rceil\cdot \left\lceil\log\left(\frac{d}{\delta}\right)\right\rceil,
    \end{align*} we can reconstruct the global model $U$ by an estimate $\hat{U}_{\it CMS}=\mc D_{\it CMS}(\mc E_{\it CMS}(X^n))$ with \begin{equation*}
        \pr(\|U-\hat{U}\|_1\leq D)\geq 1-\delta. 
    \end{equation*}
\end{lemma}
\begin{proof}
    In each round $t$, the server reconstructs the $j$-th coordinate of $X_i$, $i\in[n]$, by \begin{align}
    \label{eq::cms_larger_estimate}
        \hat{X}_i^t(j) &=Y_i^t(h^t(j))\notag\\
        &=\sum_{s=1}^d \mathds{1}(h^t(s)=h^t(j))X_i(s)\notag\\
        &=\mathds{1}(h^t(j)=h^t(j))X_i(j) + \sum_{s\neq j}\mathds{1}(h^t(s)=h^t(j))X_i(s)\notag\\
        &=X_i(j)+\sum_{s\neq j}\mathds{1}(h^t(s)=h^t(j))X_i(s)\notag\\
        &\geq X_i(j).
    \end{align} Since $\bb E[\mathds{1}(h^t(s)=h^t(j))]=1/w$ for any $s\neq j$ and $\hat{X}_i^t(j)\geq X_i(j)$ for all $t\in [T]$, we have \begin{align*}
    \bb E[|\hat{X}_i^t(j)-X_i(j)|]&=\bb E[\hat{X}_i^t(j)-X_i(j)]\\
    &=\bb E\left[X_i(j)+\sum_{s\neq j}\mathds{1}(h^t(s)=h^t(j))X_i(s)-X_i(j)\right]\\
    &=\bb E\left[\sum_{s\neq j}\mathds{1}(h^t(s)=h^t(j))X_i(s)\right]\\
    &=\sum_{s\neq j} \bb E[\mathds{1}(h^t(s)=h^t(j))]X_i(s)\\
    &= \frac{1}{w}\sum_{s\neq j}X_i(s)\leq \frac{\|X_i\|_1}{w} \leq \frac{kq}{w}.
\end{align*}
    By Markov's inequality, \begin{align*}
    \pr\left(\hat{X}_i^t(j)-X_i(j)>\frac{D}{2k}\right)&\leq \frac{\bb E[\hat{X}_i^t(j)-X_i(j)]}{D/(2k)}\leq \frac{kq/w}{D/(2k)}= \frac{2k^2q}{wD}.
\end{align*}
Note that for fixed $i\in[n]$ and $j\in[d]$, $(X_i^t(j)$, $t\in[T])$ are i.i.d.
Then we have \begin{align*}
    \pr(|\hat{X}_i(j)-X_i(j)|>D/(2k)) &= \pr(\hat{X}_i(j)-X_i(j)>D/(2k))\\
    &= \pr\left(\min_{t\in[T]} \hat{X}_i^t(j)-X_i(j)>D/(2k)\right)\\
    &= \pr(\forall t\in[T],\, \hat{X}_i^t(j)-X_i(j)>D/(2k))\\
    &= \prod_{t=1}^T \pr(|\hat{X}_i^t(j)-X_i(j)|>D/(2k)) \leq \left(\frac{2k^2q}{wD}\right)^T.
\end{align*}
Taking the union bound over $j\in[d]$ and letting $w=\lceil \frac{4k^2q}{D}\rceil$, \begin{align}
\label{eq::event_a_prob}
     \pr\left(\max_{j\in [d]}|\hat{X}_i(j)-X_i(j)|>\frac{D}{2k}\right)\leq d\cdot \left(\frac{2k^2q}{wD}\right)^T \leq d\cdot 2^{-T}.
\end{align}
Similar to CS, we consider the event \begin{equation*}
    \mc A := \left\{\max_{j\in[d]}|\hat{X}_i(j)-X_i(j)|\leq \frac{D}{2k}\right\}
\end{equation*} and threshold out every coordinate $j$ with $|\hat{X}_i(j)|\leq D/(2k)$. Denote the new estimation as $\Tilde{X}_i(j)$. Using exactly the same argument following \eqref{eq::threshold}, we can obtain that for all $i\in[n]$, given the event $\mc A$, for any $i\in[n]$,\[
\|\Tilde{X}_i-X_i\|_1\leq D.
\]
Letting $T=\lceil\log\left(\frac{d}{\delta}\right)\rceil$ and \[
\hat{U} = \frac{1}{n}\sum_{i=1}^n \Tilde{X}_i,
\]
we have
\begin{align*}
\pr(\|\hat{U}-U\|_1\leq D) \geq \pr\left(\max_{i\in[n]}\|\Tilde{X}_i-X_i\|_1\leq D\right) \geq \pr(\mc A)\utag{i}{\geq}1-d\cdot 2^{-T}\geq 1-\delta,
\end{align*}
where (i) holds by \eqref{eq::event_a_prob}.
\end{proof}

For each $i\in[n]$, $t\in[T]$, and $r\in[w]$, \[
0\leq Y_i^t(w) = \sum_{j=1}^d\mathds{1}(h^t(j)=r)X_i(j)\leq \sum_{j=1}^d X_i(j)\leq kq,
\]
implying that \[
Y_i^t(r)\in\{0,1,\ldots,kq\}.
\]
To transmit all $(Y_i^1,\ldots,Y_i^T)_{i\in[n]}$ to the server, $\mc C_{\it CMS}$, the communication cost of the scheme using CMS, is \begin{align}
\label{eq::cc_CM}
    n(wT\log(kq+1)) = n\cdot \left\lceil\frac{4k^2q}{D}\right\rceil\cdot \left\lceil\log\left(\frac{d}{\delta}\right)\right\rceil\cdot \lceil\log (kq+1)\rceil,
\end{align}
where $n$ is the number of clients, $w$ is the dimension of each vector $Y_i^t$, $t\in[T]$, and $\log (kq+1)$ is the number of bits required for transmitting the value of each coordinate $Y_i^t(r)$, $r\in[w]$.

The communication cost per client of CMS grows in the order $\mc O(
\frac{k^2q\log(kq)\log d}{D})$. Since $\mc C_{\it CMS}$ grows in the order $\mc O(k^2\log k)$ in the number of non-zero entries $k$, replacing the CMS input $x_i\in\mc X_{d,k,q}$ by $y(x_i)\in\mc Y(D')$ for some $D'>0$ can provide a tradeoff in the total communication cost between the distortions $D$ and $D'$ from two compression steps. This observation motivates our calculation in the next section.

\subsection{Combining Covering and Sketching Compressions}
\label{sec::combination}
One observation for the covering compression is that the output from its encoder $y\in\mc Y(D)\subseteq\mc X_{d,k,q}$ for any $D$ is also a sparse vector, which can be further compressed through sketching methods. That is, we can combine the covering and sketching methods to develop new compression schemes \begin{equation}
\label{eq::combine_cs}
    \pi_{\it cCS} = (\mc E_{\it CS}\circ \mc E_{\it cover}, \mc D_{\it CS}),
\end{equation} and \begin{equation}
\label{eq::combine_cms}
    \pi_{\it cCMS} = (\mc E_{\it CMS}\circ \mc E_{\it cover}, \mc D_{\it CMS}).
\end{equation}

Suppose $\pi_{\it cover}\in\Pi(D_1,0)$ and $\pi_{\it CS}, \pi_{\it CMS}\in\Pi(D_2,\delta)$. 
One can easily verify that both $\pi_{\it cCS}$ and $\pi_{\it cCMS}$ defined by \eqref{eq::combine_cs} and \eqref{eq::combine_cms} are in $\Pi(D_1+D_2, \delta)$ as follows. Define the successful events \[
\mc S_{\it cover}=\left\{ \left\|\mc D_{\it cover}(\mc E_{\it cover}(X^n))-\frac{1}{n}\sum_{i=1}^n X_i\right\|_1\leq D_1\right\}
\] and
\begin{align*}
\mc S_{\it CS} = \left\{\left\|\mc D_{\it CS}(\mc E_{\it CS}(X^n))-\frac{1}{n}\sum_{i=1}^n X_i\right\|_1\leq D_2\right\}.
\end{align*}
If both $\mc S_{\it cover}$ and $\mc S_{\it CS}$ occur, then the reproduced \[
\hat{U}_{\it cCS} = \mc D_{\it CS}(\mc E_{\it CS}(\mc E_{\it cover}(X_i)),i\in[n])
\] satisfies \begin{align*}
    \left\|\hat{U}_{\it cCS}-U \right\|_1 &\utag{i}{\leq} \left\|\mc D_{\it CS}(\mc E_{\it CS} (Y(X_i)),i\in[n])-\frac{1}{n}\sum_{i=1}^n Y(X_i)\right\|_1  + \left\|\frac{1}{n}\sum_{i=1}^n Y(X_i)-\frac{1}{n}\sum_{i=1}^n X_i\right\|_1\\
    &\utag{ii}{=} \left\|\mc D_{\it CS}(\mc E_{\it CS} (Y(X_i)),i\in[n])-\frac{1}{n}\sum_{i=1}^n Y(X_i)\right\|_1  + \left\|\mc D_{\it cover}(\mc E_{\it cover}(X^n))-\frac{1}{n}\sum_{i=1}^n X_i\right\|_1\\
    &\utag{iii}{\leq} D_1+D_2,
\end{align*}
where (i) holds by triangle inequality and the definition of $\mc E_{\it cover}$ in \eqref{eq::cover_encode}, (ii) holds by the definition of $\mc D_{\it cover}$ in \eqref{eq::cover_decode}, and (iii) holds by $\mc S_{\it cover}$ and $\mc S_{\it CS}$.
Therefore, \begin{align*}
    \pr(\|\hat{U}_{\it cCS}-U\|_1\leq D_1+D_2) \geq \pr (\mc S_{\it CS}\cap \mc S_{\it cover})= 1-\pr(\mc S_{\it CS}^c\cup\mc S_{\it cover}^c) \geq 1-\delta,
\end{align*}
and hence $\pi_{\it cCS}\in\Pi(D_1+D_2, \delta)$. Similarly, we see that $\pi_{\it cCMS}\in\Pi(D_1+D_2, \delta)$.

However, the communication cost of the combined scheme is not simply the sum $C(\pi_{\it cover})+C(\pi_{\it CS})$ or $C(\pi_{\it cover})+C(\pi_{\it CMS})$. We see from \eqref{eq::range_t} in the proof of Theorem~\ref{thm::yd_size} that $\mc Y(D_1)$ contains sparse vectors with at most $\lceil k-D_1/q\rceil$ non-zero entries.
Since $\mc E_{\it cover}(X_i)=Y(X_i)$ for all $i\in[n]$ are in the subset $\mc Y(D_1)$, the communication costs of the sketching methods can be computed by replacing the $k$-sparse condition by the $\ceil{k-D_1/q}$-sparse condition in Lemma~\ref{lemma::count_sketch} and Lemma~\ref{lemma::cm}. We can then upper bound the communication-accuracy tradeoff $C(D_1+D_2,\delta)$ by minimizing $\mc C_{\it CS}$ and $\mc C_{\it CMS}$ over $D_1$ and $D_2$.

\begin{theorem}
\label{thm::achieved_indvdl}
    Fix $\delta\in (0,1)$. For any $D\in[mq,(m+1)q)$, $m=0,1,\ldots,k-1$, the scheme $\pi_{\it cCS}=(\mc E_{\it CS}\circ \mc E_{\it cover}, \mc D_{\it CS})\in\Pi(D,\delta)$ achieves \begin{itemize}
        \item \begin{align*}
            \mc C(\pi_{\it cCS})\leq n\cdot \left\lceil\frac{200k^3}{3m^3q}\right\rceil\cdot \left\lceil\log\left(\frac{d}{\delta}\right)\right\rceil\cdot \lceil\log (2kq+1)\rceil
        \end{align*}
        if $m < \lceil 2k/3\rceil$;
        \item \begin{align*}
            \mc C(\pi_{\it cCS})\leq  450 n\cdot \left(k-m\right)\cdot \left\lceil\log\left(\frac{d}{\delta}\right)\right\rceil\cdot \lceil\log (6(k-m) q+1)\rceil
        \end{align*}
        if $m\geq \lceil 2k/3\rceil$.
    \end{itemize}   
    Also, the scheme $\pi_{\it cCMS}=(\mc E_{\it CMS}\circ \mc E_{\it cover}, \mc D_{\it CMS})\in\Pi(D,\delta)$ achieves \begin{itemize}
        \item \begin{align*}
            \mc C(\pi_{\it cCMS})\leq n\cdot \left\lceil\frac{4k^2}{m}\right\rceil\cdot \left\lceil\log\left(\frac{d}{\delta}\right)\right\rceil\cdot \lceil\log (kq+1)\rceil
        \end{align*}
        if $m<\lceil k/2\rceil$;
        \item \begin{align*}
            \mc C(\pi_{\it cCMS})\leq  16 n\cdot (k-m)\cdot \left\lceil\log\left(\frac{d}{\delta}\right)\right\rceil\cdot \lceil\log (2(k-m)q+1)\rceil
        \end{align*}
        if $m\geq \lceil k/2\rceil$.
    \end{itemize}
\end{theorem}
\begin{proof}
Fix an error probability $\delta\in(0,1)$ and a total distortion $D=D_1+D_2$, where $D_1,D_2\geq 0$. For $mq\leq D<(m+1)q$, $m=0,1,\ldots,k-1$, we have $\pi\in\Pi(mq,\delta)\subseteq
\Pi(D,\delta)$. When $D\geq kq$, we have 
\[\mc C(\pi_{\it cover})=0\geq C(D,0)\geq C(D,\delta)\geq 0\] 
for any $\delta\in(0,1)$. In the rest of this proof, we only consider $D=mq$ with $m=0,1,\ldots,k-1$, and divide it into $D_1=m_1q$ and $D_2=m_2q$ with $m_1+m_2=m$, $m_1,m_2\in\{0,1,\ldots,k-1\}$.

Consider the communication cost of the scheme $\pi_{\it cCS}=(\mc E_{\it CS}\circ \mc E_{\it cover}, \mc D_{\it CS})$ with $\pi_{\it CS}\in\Pi(m_2q,\delta)$ and $\pi_{\it cover}\in\Pi(m_1q,0)$. From \eqref{eq::range_t} in the proof of Theorem~\ref{thm::yd_size}, for all $x\in\mc X_{d,k,q}$, $\mc E_{\it cover}(x)$ has at most $\lceil k- D_1/q\rceil=k-m_1$ non-zero entries, i.e., $\mc E_{\it cover}(x)\in\mc X_{d,k-m_1,q}$. Hence, the inputs of $\mc E_{\it CS}$ are $(k-m_1)$-sparse vectors.
Substituting $k$ in \eqref{eq::cc_CS} by $k-m_1$, for $\pi_{\it CS}=(\mc E_{\it CS}, \mc D_{\it CS})\in\Pi(m_2q,\delta)$, the communication cost $\mc C(\pi_{\it cCS})$ is upper bounded by 
\begin{align}
\label{eq::comm_cost_comb_cs}
    &n\cdot \left\lceil\frac{200(k-m_1)^3q^2}{3(m_2q)^2}\right\rceil\cdot \left\lceil\log\left(\frac{d}{\delta}\right)\right\rceil\cdot \lceil\log (2\lceil k-D_1/q\rceil q+1)\rceil\notag\\
    &\qquad= \left\lceil\frac{200n(k-m_1)^3}{3m_2^2}\right\rceil\cdot \left\lceil\log\left(\frac{d}{\delta}\right)\right\rceil\cdot \lceil\log (2(k-m_1) q+1)\rceil.
\end{align}
Similarly, from \eqref{eq::cc_CM}, the communication cost $\mc C(\pi_{\it cCMS})$ is upper bounded by \begin{align}
\label{eq::comm_cost_comb_cms}
    &n\cdot \left\lceil\frac{4(k-m_1)^2q}{m_2q}\right\rceil\cdot \left\lceil\log\left(\frac{d}{\delta}\right)\right\rceil\cdot \lceil\log (\lceil k-D_1/q\rceil q+1)\rceil\notag\\
    &\qquad =\left\lceil\frac{4n(k-m_1)^2}{m_2}\right\rceil\cdot \left\lceil\log\left(\frac{d}{\delta}\right)\right\rceil\cdot \lceil\log ((k-m_1) q+1)\rceil
\end{align}

Recalling that the last logarithmic terms in \eqref{eq::comm_cost_comb_cs} and \eqref{eq::comm_cost_comb_cms} are the numbers of bits used in transmitting the values of the coordinates of a vector. We then focus on minimizing the dimensions of the outputs of $\mc E_{\it CS}\circ \mc E_{\it cover}$ and $\mc E_{\it CMS}\circ \mc E_{\it cover}$, i.e., \begin{equation}
\label{eq::cs_output_dim}
\left\lceil\frac{200(k-m_1)^3}{3m_2^2}\right\rceil
\end{equation} and \begin{equation}
\label{eq::cms_ouput_dim}
\left\lceil\frac{4(k-m_1)^2q}{m_2}\right\rceil,
\end{equation}
respectively, over $m_2\leq m$ and $m_1=m-m_2$, where $m,m_2\in\{0,1,\ldots,k-1\}$. 

By substituting $m_1$ by $m-m_2$, our targets become minimizing \begin{align}
\label{eq::comb_cs_cc_1}
\frac{(k-m+m_2)^3}{m_2^2} = \left(k-m+m_2\right)^3\cdot \left(\frac{1}{m_2}\right)^2= \left(k-m+m_2\right)^3\cdot \left(m_2^{-\frac{2}{3}}\right)^3= \left(\frac{k-m}{m_2^{2/3}}+m_2^{\frac{1}{3}}\right)^3
\end{align} for $\pi_{\it cCS}=(\mc E_{\it CS}\circ \mc E_{\it cover}, \mc D_{\it CS})$, and \begin{align}
\label{eq::comb_cms_cc_1}
\frac{(k-m+m_2)^2}{m_2} = \left(k-m+m_2\right)^2\cdot \left(\frac{1}{m_2}\right)= \left(k-m+m_2\right)^2\cdot \left(m_2^{-\frac{1}{2}}\right)^2= \left(\frac{k-m}{m_2^{1/2}}+m_2^{\frac{1}{2}}\right)^2
\end{align} for $\pi_{\it cCMS}=(\mc E_{\it CMS}\circ \mc E_{\it cover}, \mc D_{\it CMS})$.

Observe that \eqref{eq::comb_cs_cc_1} can be expressed as $f(x)=(cx^{-\frac{2}{3}}+x^{\frac{1}{3}})^3$ with $x=m_2$.
By elementary computation, it can be shown that \eqref{eq::comb_cs_cc_1} attains the minimum at \begin{equation*}
    m_2^* = \begin{cases}
        2(k-m) &\mbox{for}\ m\geq \lceil 2k/3\rceil \\
        m &\mbox{for}\ m<\lceil 2k/3 \rceil
    \end{cases}.
\end{equation*}
Accordingly, $m_1^* = m-m_2^*$ is given by \begin{equation*}
    m_1^* = \begin{cases}
        3m-2k &\mbox{for}\ m\geq \lceil 2k/3\rceil \\
        0 &\mbox{for}\  m<\lceil 2k/3 \rceil
    \end{cases}.
\end{equation*}
For $m<\lceil 2k/3\rceil$, by substituting $(m_1, m_2)$ in \eqref{eq::comm_cost_comb_cs} by $(m_1^*, m_2^*)=(0,m)$, the communication cost $\mc C(\pi_{\it cCS})$ is upper bounded by \[
    n\cdot \left\lceil\frac{200k^3}{3m^3q}\right\rceil\cdot \left\lceil\log\left(\frac{d}{\delta}\right)\right\rceil\cdot \lceil\log (2kq+1)\rceil.
\] 
For $m\geq \lceil 2k/3\rceil$, upon substituting $m_2$ by $m_2^*=2(k-m)$, \eqref{eq::comb_cs_cc_1} becomes \begin{align*}
    \frac{(k-m+m_2^*)^3}{m_2^{*2}}=\frac{(3(k-m))^3}{(2(k-m))^2}=\frac{27(k-m)}{4}.
\end{align*}
Then following \eqref{eq::cs_output_dim}, the dimension of the output of $\mc E_{\it CS}\circ \mc E_{\it cover}$ is upper bounded by $450(k-m)$.

Moreover, substituting $m_1^* = 3m-2k$ in the last logarithmic term in \eqref{eq::comm_cost_comb_cs}, we can upper bound $\mc C(\pi_{\it cCS})$ by
\begin{align*}
450 n\cdot (k-m)\cdot \left\lceil\log\left(\frac{d}{\delta}\right)\right\rceil\cdot \lceil\log (6(k-m)q+1)\rceil
\end{align*}
for $m\geq \lceil 2k/3\rceil$.

Applying the same process to $\pi_{\it cCMS}$, \eqref{eq::comb_cms_cc_1} attains the minimum at \begin{equation*}
    m_2^{*\prime} = \begin{cases}
        k-m &\mbox{for}\ m\geq \lceil k/2\rceil\\
        m &\mbox{for}\ m<\lceil k/2 \rceil
    \end{cases}.
\end{equation*}
Accordingly, $m_1^{*\prime} = m-m_2^{*\prime}$ is given by \begin{equation*}
    m_1^{*\prime} = \begin{cases}
        2m - k &\mbox{for}\  m\geq \lceil k/2\rceil\\
        0 &\mbox{for}\ m<\lceil k/2\rceil
    \end{cases}.
\end{equation*}
For $m<\lceil k/2\rceil$, by substituting $(m_1,m_2)$ in \eqref{eq::comm_cost_comb_cms} by $(m_1^{*\prime}, m_2^{*\prime})=(0,m)$ , the communication cost $\mc C(\pi_{\it cCMS})$
is upper bounded by \[
n\cdot \left\lceil\frac{4k^2}{m}\right\rceil\cdot \left\lceil\log\left(\frac{d}{\delta}\right)\right\rceil\cdot \lceil\log (kq+1)\rceil.
\] 
For $m\geq \lceil k/2\rceil$, by substituting $m_2$ by $m_2^{*\prime}$, \eqref{eq::comb_cms_cc_1} becomes \begin{align*}
    \frac{(k-m+m_2^{*\prime})^2}{m_2^{*\prime}}=\frac{4(k-m)^2}{(k-m)} = 4(k-m).
\end{align*}
Then following \eqref{eq::cms_ouput_dim}, the dimension of the output of $\mc E_{\it CMS}\circ \mc E_{\it cover}(x)$ is upper bounded by $16(k-m)$.

Moreover, substituting $m_1^{*\prime}=m-m_2^{*\prime}=2m-k$ in the last logarithmic term in \eqref{eq::comm_cost_comb_cms}, we can upper bound $\mc C(\pi_{\it cCMS})$ by
\begin{align*}
16n\cdot (k-m) \cdot \left\lceil\log\left(\frac{d}{\delta}\right)\right\rceil\cdot \lceil\log ( 2(k-m) q+1)\rceil
\end{align*}
when $m\geq \lceil k/2\rceil$.

\end{proof}
The numerical results of the upper bounds on $\mc C(\pi_{\it CMS})$ and $\mc C(\pi_{\it cCMS})$ are illustrated in Fig.~\ref{fig::results}. The results about $CS$ and its combination with the covering scheme are not included in the illustration since they are much larger than those of CMS. We first observe that combining the sketching methods with a covering code can potentially reduce the communication costs as the distortion crosses the threshold $\frac{q(k+1)}{2}$.
The communication cost of both $\pi_{\it CMS}$ and $\pi_{\it cCMS}$ under small error probability $\delta$ are larger than that of using the covering scheme only. 
As we mentioned in Section~\ref{sec::related_works}, it is necessary to employ randomness and such an increase in communication cost is inevitable and should not be regarded as a ``waste" of communication resources.

\begin{figure}
    \centering	
    \includegraphics[width=0.5\linewidth]{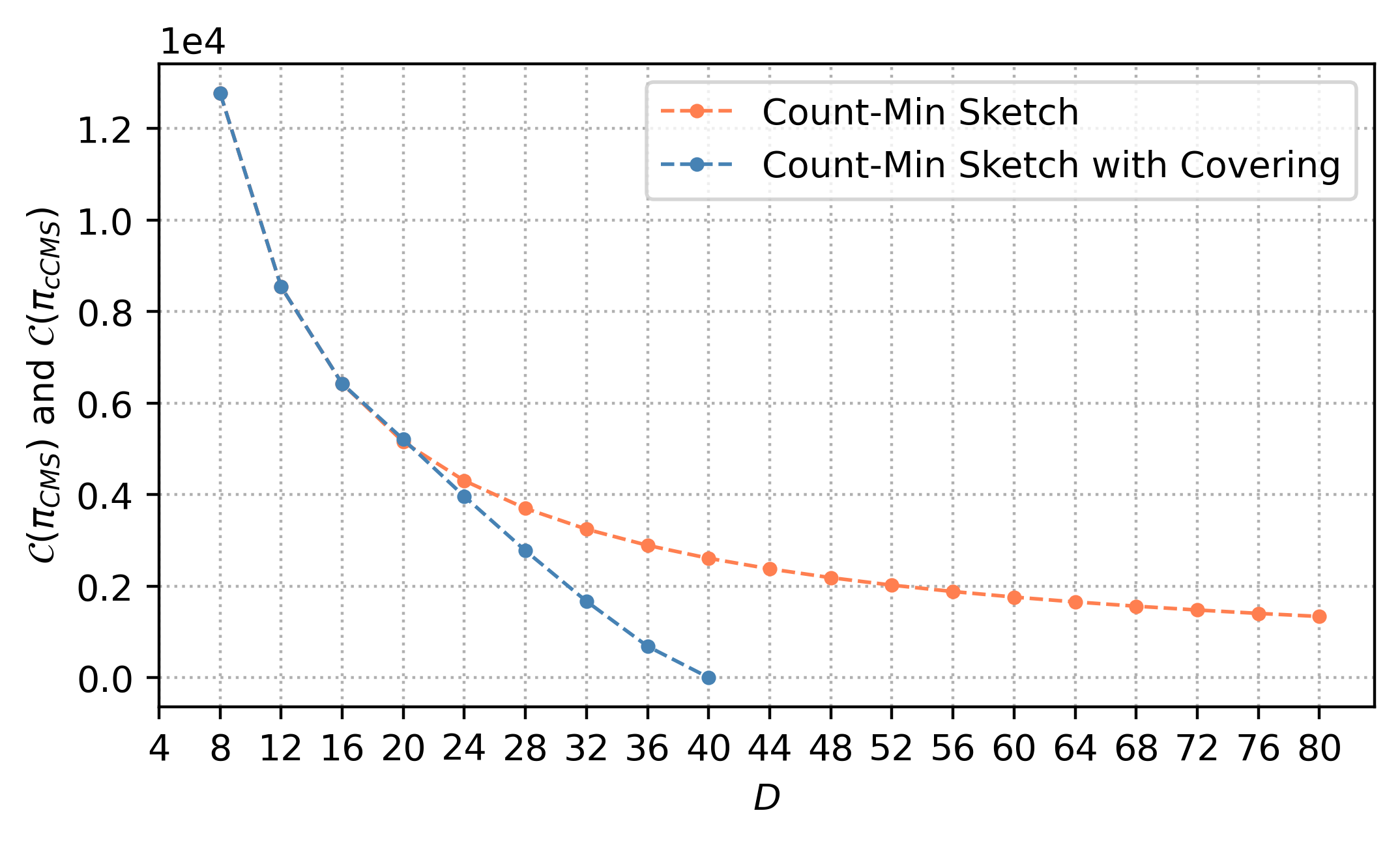}
    \caption{Numerical illustration of the upper bounds on the achievable communication costs using CMS and CMS combined with the covering scheme. In this illustration, we use $n=5$, $k=10$, $d=2^{10}$, $q=4$, $\delta=0.1$, and the values of $D$ are multiples of $q$. }
  \label{fig::results}
\end{figure}

\section{Converse using $\mathrm{f}$-Information Measure}
\label{sec::converse}
In this section, we obtain a lower bound on the communication cost $C(D,\delta)$ in terms of the distortion $D$ and error probability $\delta$. This lower bound applies to the case when the global model $U$ is a random variable with an underlying distribution $P_U$ with the constraint \begin{equation*}
    \pr(\|\hat{U}-U\|_1>D)<\delta
\end{equation*}
defined as \eqref{eq::accuracy}.
In the previous sections, the global model $U$ is treated as a constant vector $U=u$ since we require our schemes to apply to arbitrarily distributed $X_i$'s. The corresponding accuracy requirement on $\hat{U}=\mathcal{D}(\mathcal{E}(X^n))$ is \begin{equation}
\label{eq::accracy_requirement_constant_u}
\pr(\|\hat{U}-u\|_1>D)<\delta \text{ for any } u\in\mathcal{U},
\end{equation}
which implies that \eqref{eq::accuracy} is satisfied since
\begin{align*}
    \pr(\|\hat{U}-U\|_1>D) = \sum_{u\in\mc U}P_U(u)\pr(\|\hat{U}-u\|_1>D)< \sum_{u\in\mc U}P_U(u)\cdot \delta=\delta.
\end{align*}

For a random variable $U$, one can first relate the achievable communication cost with $I(U;\hat{U})$.
Then a natural choice of tool is the generalized Fano's inequality developed in \cite{zhang2013information}, which can lower bound $I(U;\hat{U})$ under the accuracy requirement. However, we show that by replacing Shannon's information measure with 
the f-information measure to obtain a tighter lower bound on the achievable communication costs under certain distortion $D$ and error probability $\delta$, which can strictly improve the converse provided by the generalized Fano's inequality in \cite{zhang2013information}.

The $\mathrm{f}$-divergence by Ali-Silvey~\cite{ali1966general} and Csiszar~\cite{csiszar1967information} is defined as 
\begin{align}
    \label{eq::f_div}
    D_\mathrm{f}(P\Vert Q) = \mathbb{E}_Q\left[\mathrm{f}\left(\frac{\mathrm{d}P}{\mathrm{d}Q}(X)\right)\right],
\end{align}
where $\mathrm{f}:(0,+\infty)\rightarrow \mathbb{R}$ is a convex function with $\mathrm{f}(1)=0$ and $\mathrm{f}(0)\triangleq \mathrm{f}(0^+)$, $P$ and $Q$ are two probability distributions on a measurable space $(\mathcal{X},\sigma(\mathcal{X}))$ with $P\ll Q$, and $\frac{\mathrm{d}P}{\mathrm{d}Q}$ denotes the Radon-Nikodym derivative. 

There have been multiple proposals for defining mutual f-information based on the f-divergence. Here we employ the one referred to as $I_\mathrm{f}^{\it CKZ}(X;Y)$ in \cite{masiha2023f}, defined as \begin{equation}
\label{eq::f_mutual_info}
    I_\mathrm{f}^{\it CKZ}(A;B) := D_\mathrm{f}(P_{AB}\Vert P_A\times P_B).
\end{equation}
Correspondingly, the f-entropy of a random variable $X$ on a finite discrete alphabet $\mathcal{X}$ is defined as \begin{align}
    \label{eq::f_entropy}
    H_\mathrm{f}^{\it CKZ}(X) \triangleq I_\mathrm{f}^{\it CKZ}(X;X) = \mathrm{f}(0)\left(1-\sum_{x\in\mc X}p_x^2\right) + \sum_{x\in\mc X}p_x^2\mathrm{f}\left(\frac{1}{p_x}\right).
\end{align} 

The following properties holds for the mutual f-information defined above. The proofs are omitted here.
\begin{proposition}[\cite{masiha2023f}]
\label{prop::renyi}
    Let $X$ and $Y$ be random variables defined on a finite set $\mc X$. Then \begin{enumerate}[label=(\roman*)]
        \item the function $p_X\mapsto H_\mathrm{f}(X)$ is maximized by the uniform distribution if $\mathrm{f}'''(x)\leq 0$ for $x>0$;
        \item $I_\mathrm{f}^{\it CKZ}(X;Y)\leq H_\mathrm{f}^{\it CKZ}(Y)$;
        \item let $P_X\rightarrow P_Y$ and $Q_X\rightarrow Q_Y$, where $P_Y$ and $Q_Y$ are obtained from $P_X$ and $Q_X$, respectively, through a fixed transition probability $A_{Y|X}$. Then \begin{equation*}
            D_\mathrm{f}(P_Y||Q_Y)\leq D_\mathrm{f}(P_X||Q_X).
        \end{equation*}
        \item $I_\mathrm{f}^{\it CKZ}(X;Z)\leq I_\mathrm{f}^{\it CKZ}(X;Y)$ if the Markov chain $X-Y-Z$ holds.
    \end{enumerate}
\end{proposition}
Since $I_\mathrm{f}^{\it CKZ}(X;Y)$ is the mutual f-information employed in the rest of this section, we will ignore the superscript CKZ and write it as $I_\mathrm{f}(X;Y)$. Likewise, we will write $H_\mathrm{f}^{\it CKZ}(X)$ as $H_\mathrm{f}(X)$.

Consider our communication-accuracy tradeoff problem for sparse aggregation.
For any scheme $\pi=(\mc E,\mc D)$ with $\mc E(X)\in\mc Y$, we have $\mc E(X^n)\in\mc Y^n$ and, if $f$ satisfies $f'''(x)\leq 0$ for $x>0$, then
\begin{align}
\label{eq::lower_bd_f_yn}
    \frac{\mathrm{f}(|\mc Y|^n)-\mathrm{f}(0)}{|\mc Y|^n} + \mathrm{f}(0) \utag{a}{\geq}  H_\mathrm{f}(\mc E(X^n))\utag{b}{\geq} I_\mathrm{f}(X^n;\mc E(X^n)),
\end{align} 
where \uref{a} holds by Proposition~\ref{prop::renyi} (i) with $p_x$ replaced by $1/|\mc Y^n|$ in \eqref{eq::f_entropy}, and \uref{b} holds by Proposition~\ref{prop::renyi} (ii). In the following theorem, we lower bound the mutual f-information $I_\mathrm{f}(X^n;\mc E(X^n))$ under the distortion constraint, and then obtain a lower bound on $C(\pi)$ for any $\pi\in\Pi(D,\delta)$.

\begin{theorem}
\label{thm::lower_bd_f_div}
    Let \[
    P_{\max}(D) = \max_{u'\in\mc U} \, \sum_{u\in\mc U}P_U(u)\cdot\mathds{1}(\|u-u'\|_1\leq D).
    \]
    For any $D\geq 0$, $0\leq \delta\leq 1-P_{\max}(D)$ (or equivalently $P_{\max}(D)\leq 1-\delta\leq 1$), 
    and any convex function $\mathrm{f}:(0,+\infty)\rightarrow\mathbb{R}$ with \begin{itemize}
        \item $\mathrm{f}(1)=0$, and
        \item $\mathrm{f}'''(x)\leq 0$ for any $x>0$,
    \end{itemize}
    the alphabet size $|\mc Y|$ for $\mc E(X)\in\mc Y$ of any $\pi=(\mc E,\mc D)\in\Pi(D,\delta)$ satisfies \begin{align}
        \frac{\mathrm{f}(|\mc Y|^n)-\mathrm{f}(0)}{|\mc Y|^n} \geq \left(1-P_{\max}(D)\right)\cdot\mathrm{f}\left(\frac{\delta}{1-P_{\max}(D)}\right) + P_{\max}(D)\cdot\mathrm{f}\left(\frac{1-\delta}{P_{\max}(D)}\right)-\mathrm{f}(0),
        \label{eq::converse_f_div}
    \end{align}
    where $\mathrm{f}(0)\triangleq \mathrm{f}(0^+)$.
    
\end{theorem}
\begin{proof}
    First, by Proposition~\ref{prop::renyi} (iv) and the Markov chain $U-X^n-\mc E(X^n)$ for any $\mc E:\mc X\rightarrow\mc Y$, we have \begin{align*}
        I_\mathrm{f}(X^n;\mc E(X^n)) \geq I_\mathrm{f}(U;\mc E(X^n)).
    \end{align*}
    Since $\hat{U}$ is a function of $\mc E(X^n)$, the Markov chain $U-\mc E(X^n)-\hat{U}$ holds. Then by Proposition~\ref{prop::renyi} (iv), we have \[
    I_\mathrm{f}(U;\mc E(X^n)) \geq I_\mathrm{f}(U;\hat{U}).
    \]
    
    Consider in Proposition~\ref{prop::renyi} (iii) the transition probability induced by the data processor \begin{equation*}
        (U, \hat{U})\mapsto\mathds{1}(\|U-\hat{U}\|_1\leq D).
    \end{equation*} Then $P_{U,\hat{U}}$ is mapped to $\mathrm{Ber}(q)$, where \begin{align*}
    q := \sum_{u,u'\in\mc U}P_{U,\hat{U}}(u,u')\mathds{1}(\|u-u'\|_1\leq D)= \pr(\|U-\hat{U}\|_1\leq D),
    \end{align*}
    and $P_UP_{\hat{U}}$ is mapped to $\mathrm{Ber}(p)$, where \begin{align*}
        p:=\sum_u\sum_{u'} P_U(u)P_{\hat{U}}(u') \mathds{1}(\|u-u'\|_1 \leq D)\leq \max_{u'\in\mc U}\, \sum_u P_U(u)\cdot \mathds{1}(\|u-u'\|_1 \leq D) = P_{\max}(D).
    \end{align*}  
    The accuracy constraint \eqref{eq::accuracy} requires that $q\geq 1-\delta$. Since the achievable communication cost increases as the error probability $\delta$ decreases,
    we set $q=\pr(\|U-\hat{U}\|_1)$ to be $1-\delta$ in order to obtain a lower bound.
    
    Using Proposition~\ref{prop::renyi} (iii), we then lower bound $I_\mathrm{f}(U;\hat{U})$ by
    \begin{align}
    \label{eq::f_div_binary}
        I_\mathrm{f}(U;\hat{U}) = D_\mathrm{f}(P_{U,\hat{U}}\Vert P_UP_{\hat{U}})\geq D_\mathrm{f}(\mathrm{Ber}(q)\Vert \mathrm{Ber}(p)) = D_\mathrm{f}(\mathrm{Ber}(1-\delta)\Vert \mathrm{Ber}(p)).
    \end{align}
    Next, we show that $D_\mathrm{f}(\mathrm{Ber}(1-\delta)\Vert \mathrm{Ber}(p))$ is non-increasing in $p$, implying that \[
    D_\mathrm{f}(\mathrm{Ber}(1-\delta)\Vert \mathrm{Ber}(p)) \geq D_\mathrm{f}(\mathrm{Ber}(1-\delta)\Vert \mathrm{Ber}(P_{\max}(D))).
    \]
    First observe that, since $\delta\leq 1-P_{\max}(D)\leq 1-p$, we have \[
    \frac{\delta}{1-p} \leq \frac{1-\delta}{p}.
    \]
    Write \eqref{eq::f_div_binary} as \begin{align*}
        g(p) \triangleq p\cdot\mathrm{f}\left(\frac{1-\delta}{p}\right) + (1-p)\cdot \mathrm{f}\left(\frac{\delta}{1-p}\right).
    \end{align*}
    Then we have \begin{align*}
    g'(p)&= \frac{\delta}{1-p}\mathrm{f}'\left(\frac{\delta}{1-p}\right) - \frac{1-\delta}{p}\mathrm{f}'\left(\frac{1-\delta}{p}\right) + \left(\mathrm{f}\left(\frac{1-\delta}{p}\right) - \mathrm{f}\left(\frac{\delta}{1-p}\right) \right)\\
    & \utag{a}{\leq} \frac{\delta}{1-p}\mathrm{f}'\left(\frac{\delta}{1-p}\right) - \frac{1-\delta}{p}\mathrm{f}'\left(\frac{1-\delta}{p}\right)  + \mathrm{f}'\left(\frac{1-\delta}{p}\right)\left(\frac{1-\delta}{p}-\frac{\delta}{1-p}\right)\\
    & = \frac{\delta}{1-p}\cdot \left(\mathrm{f}'\left(\frac{\delta}{1-p}\right)-\mathrm{f}'\left(\frac{1-\delta}{p}\right)\right) \utag{b}{\leq} 0,
    \end{align*}
    where (a) holds since for the convex function $f(t)$, $f(a)-f(b)\geq f'(b)(a-b)$ for $a,b\geq 0$, and (b) holds since $f'$ is non-decreasing by the convexity of $f$. 
    
    Therefore, $D_\mathrm{f}(\mathrm{Ber}(1-\delta)\Vert \mathrm{Ber}(p))$ is non-increasing in $p$ and we can lower bound \eqref{eq::f_div_binary} by $g(P_{\max}(D))$, or \[
    P_{\max}(D)\cdot\mathrm{f}\left(\frac{1-\delta}{P_{\max}(D)}\right)+\left(1-P_{\max}(D)\right)\cdot\mathrm{f}\left(\frac{\delta}{1-P_{\max}(D)}\right).
    \]
    Replacing $I_\mathrm{f}(X^n;\mc E(X^n))$ in \eqref{eq::lower_bd_f_yn} by this lower bound, we have \begin{align}
        \frac{\mathrm{f}(|\mc Y|^n)-\mathrm{f}(0)}{|\mc Y|^n} + \mathrm{f}(0)\geq P_{\max}(D)\cdot\mathrm{f}\left(\frac{1-\delta}{P_{\max}(D)}\right)  +\left(1-P_{\max}(D)\right)\cdot\mathrm{f}\left(\frac{\delta}{1-P_{\max}(D)}\right).\notag
    \end{align}

\end{proof}

Now, consider the convex function $\mathrm{f}(t)=t\log t$ with $\mathrm{f}(0)\triangleq 0$, which also satisfies the conditions in Theorem~\ref{thm::lower_bd_f_div}. With this choice of $\mathrm{f}$, $D_\mathrm{f}(P\Vert Q)$ becomes the Kullback-Leibler divergence (KL divergence). Substituting this $\mathrm{f}$ into \eqref{eq::converse_f_div}, we obtain \begin{align}
\label{eq::converse_shannon_im}
C(D,\delta) \geq n\log |\mc Y|\geq \delta\log\frac{\delta}{1-P_{\max}(D)} + (1-\delta)\log \frac{1-\delta}{P_{\max}(D)},
\end{align}
which is usually referred to as the generalized Fano's inequality \cite{zhang2013information} \cite{scarlett2019introductory}. 

Next, consider the function $\mathrm{f}(t)=\frac{t^\lambda-1}{\lambda-1}$ for $\lambda>1$, which satisfies the requirements in Theorem~\ref{thm::lower_bd_f_div}. Noting that $\mathrm{f}(0)=-\frac{1}{\lambda - 1}$ and following \eqref{eq::converse_f_div}, we then have \begin{align*}
    &\frac{\frac{1}{\lambda - 1}(|\mc Y^n|^\lambda-1)-(-\frac{1}{\lambda - 1})}{|\mc Y^n|}\\
    &\qquad \geq (1-P_{\max}(D))\cdot\frac{1}{\lambda - 1} \left(\frac{\delta^\lambda}{(1-P_{\max}(D))^\lambda}-1\right) + P_{\max}(D)\cdot \frac{1}{\lambda - 1}\left(\frac{(1-\delta)^\lambda}{P_{\max}(D)^\lambda}-1\right) - \frac{-1}{\lambda-1}\\
    &\qquad = \frac{1}{\lambda - 1} \bigg((1-P_{\max}(D))\left(\frac{\delta}{1-P_{\max}(D)}\right)^\lambda  + P_{\max}(D)\left(\frac{1-\delta}{P_{\max}(D)}\right)^\lambda\bigg).
\end{align*}
That is, \begin{align*}
|\mc Y^n|\geq \bigg((1-P_{\max}(D))\left(\frac{\delta}{1-P_{\max}(D)}\right)^\lambda  + P_{\max}(D)\left(\frac{1-\delta}{P_{\max}(D)}\right)^\lambda\bigg)^{\frac{1}{\lambda-1}}.
\end{align*}
Taking logarithm on both sides, for any alphabet $\mc Y$ of $\mc E(X)$ for $\pi=(\mc E,\mc D)\in\Pi(D,\delta)$, we have \begin{align}
\label{eq::converse_chi_2}
C(D,\delta)\geq n\log|\mc Y| \geq \frac{1}{\lambda - 1}\log\bigg((1-P_{\max}(D))\left(\frac{\delta}{1-P_{\max}(D)}\right)^\lambda  + P_{\max}(D)\left(\frac{1-\delta}{P_{\max}(D)}\right)^\lambda\bigg).
\end{align}

\begin{figure}
    \centering	
    \includegraphics[width=0.5\linewidth]{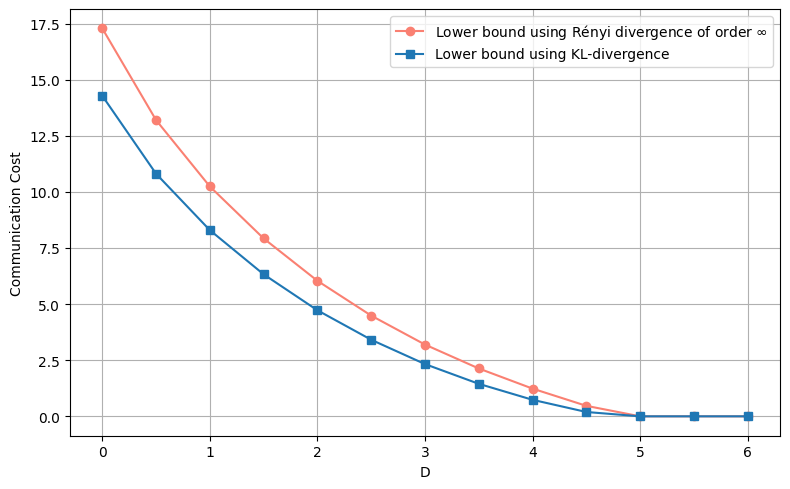}
    \caption{Comparison between the converse results given by \eqref{eq::chi_2_lb_uniform_u} and \eqref{eq::mi_lb_uniform_u} under uniform distribution assumption of $U$ with $\delta=0.15$, $n=2$, $d=11$, $k=3$, and $q=2$.}
  \label{fig::lower_bd_compare}
\end{figure}

To compare the lower bounds in \eqref{eq::converse_shannon_im} and \eqref{eq::converse_chi_2}, we use the R\'enyi divergence of order $\lambda$, defined as \[
D_\lambda(P\Vert Q) = \frac{1}{\lambda - 1}\log \mathbb{E}_Q\left[\left(\frac{\mathrm{d}P}{\mathrm{d}Q}\right)^\lambda\right],
\]
and $D_1(P\Vert Q)=D_{\it KL}(P\Vert Q)$ by taking $\lambda\rightarrow 1$.
One can observe that the right-hand-side of \eqref{eq::converse_shannon_im} and \eqref{eq::converse_chi_2} are $D_1(\mathrm{Ber}(1-\delta)\Vert \mathrm{Ber}(P_{\max}(D)))$ and $D_\lambda(\mathrm{Ber}(1-\delta)\Vert \mathrm{Ber}(P_{\max}(D)))$, respectively. Since $D_\lambda(P\Vert Q)$ is non-decreasing in $\lambda$, and often strictly so \cite{van2014renyi}, the lower bound provided in \eqref{eq::converse_chi_2} improves the one in \eqref{eq::converse_shannon_im}, regardless of the values of $\delta$ and $P_{\max}(D)$. Moreover, as $\lambda\rightarrow\infty$, \begin{align*}
C(D,\delta) \geq \lim_{\lambda\rightarrow\infty}D_\lambda(\mathrm{Ber}(1-\delta)\Vert \mathrm{Ber}(P_{\max}(D)))= \log \, \max\left\{\frac{1-\delta}{P_{\max}(D)}, \frac{\delta}{1-P_{\max}(D)}\right\},
\end{align*}
which is $\log \frac{1-\delta}{P_{\max}(D)}$ under the assumption in Theorem~\ref{thm::lower_bd_f_div}.

To numerically evaluate the lower bounds, we still need to determine the value of $P_{\max}(D)$, which depends on the distribution of $U$. Since the lower bound applies to an arbitrarily distributed $U$, for the convenience of comparison, we assume in the rest of this section that $U$ is uniformly distributed on $\mc U$, implying that $P_{\max}(D) = N_{\max}(D)/|\mc U|$, where \[
N_{\max}(D) = \max_{u'\in\mc U}|\{u\in\mc U:\|u-u'\|_1\leq D\}|.
\]
Then, by $\lambda\rightarrow \infty$ we have \begin{equation}
\label{eq::chi_2_lb_uniform_u}
    C(D,\delta)\geq \log \frac{(1-\delta)|\mc U|}{P_{\max}(D)},
\end{equation} and \eqref{eq::converse_shannon_im} becomes \begin{align}
\label{eq::mi_lb_uniform_u}
    C(D,\delta)\geq  \delta\log \frac{\delta|\mc U|}{|\mc U|-N_{\max}(D)}  + (1-\delta)\log \frac{(1-\delta)|\mc U|}{N_{\max}(D)}.
\end{align}

Fig.~\ref{fig::lower_bd_compare} illustrates the lower bounds on $C(D,\delta)$ numerically for $\delta=0.15$ using $\lambda=1$ and $\lambda=\infty$. It is observed that $D_\infty(P\Vert Q)$ provides a tighter lower bound under the same $D$ and $\delta$. Small parameters are chosen so that the ratio $N_{\max}(D)/|\mc U|$ can be computed exactly.\footnote{The complexity in evaluating the lower bounds in \eqref{eq::chi_2_lb_uniform_u} and \eqref{eq::mi_lb_uniform_u} is mainly on obtaining $N_{\max}(D)$ and $|\mathcal{U}|$ through counting.} 

Note that when $\delta=0$, the lower bounds on $C(D,0)$ given in \eqref{eq::chi_2_lb_uniform_u} and \eqref{eq::mi_lb_uniform_u} coincide with each other, which is \begin{equation}
    C(D,0) \geq \log \frac{|\mathcal{U}|}{N_{\max}(D)}.
    \label{eq::uniform_u_lb}
\end{equation} This is exactly the lower bound used for comparison in Fig~\ref{fig::delta_0_compare}.

Furthermore, we can prove an explicit lower bound on $C(D,0)$ in terms of $d,k,q$, and $D$. When $q=1$, this lower bound is close to the achievable communication cost using the covering scheme.
\begin{theorem}
\label{thm::converse_approx}
    For any distortion $D\geq 0$,  \begin{align}
        C(D,0) \geq \left(\log \frac{\sum_{c=0}^{k}\binom{d}{nc}q^{nc}}{\sum_{h=0}^{\lfloor nD\rfloor}\sum_{i=\lceil h/nq\rceil}^{\min\{h,2nk\}}\binom{d}{i} \cdot 2^i \cdot N(h,i,nq)}\right)^+,
        \label{eq::converse_approx}
    \end{align}
    where \[
    N(h,i,nq):=\sum_{j=0}^{\min\{i,\lfloor\frac{h-i}{nq}\rfloor\}} (-1)^j\binom{i}{j}\binom{h -jnq-1}{i-1}.
    \]
    \label{thm::counting_converse}
\end{theorem}
\begin{proof}
    Since when $\delta=0$, both \eqref{eq::chi_2_lb_uniform_u} and \eqref{eq::mi_lb_uniform_u} reduce to \eqref{eq::uniform_u_lb}, we further relax this lower bound by providing a lower bound on $|\mathcal{U}|$ and an upper bound on $N_{\max}(D)$.

    First, we lower bound $|\mc U|$.
    Because $d\gg nk$, we can choose the supports of the $n$ vectors $x_1,\ldots,x_n$ to be disjoint. Suppose each vector $x_i$ has $c$, $c\leq k$, non-zero coordinates. Then no coordinate is ever added with another when $x_1,\ldots,x_n$ are aggregated into $u$, which has exactly $nc$ non-zero coordinates taking values in $\{1/n,2/n,\ldots,q/n\}$. That is, \[
    \mc U \supseteq \{u\in\{0,1/n,\ldots,q/n\}^d: \|u\|_0=nc, c\in[0:k]\},
    \] implying that \begin{equation}
    |\mathcal{U}| \geq \sum_{c=0}^k \binom{d}{nc}q^{nc}.
    \label{eq::lower_bd_U}
    \end{equation}

    Next, we upper bound $N_{\max}(D)$. Consider $\{u\in\mathcal{U}: \|u-u'\|_1\leq D\}$ for a fixed $u'\in\mathcal{U}$. Denote \[
    \Delta := n(u-u')\] with
    \[\|\Delta\|_0=i \quad \text{and} \quad \|\Delta\|_1=h,
    \] 
    where $i\leq 2nk$ and $h\in\{0,1,\ldots,\lfloor nD\rfloor\}$.
    Since $|\Delta(i)|\geq 1$ if $\Delta(i)\neq 0$ and $|\Delta(i)|\leq nq$ for all $i$,
    we have
    \[
    i = \|\Delta\|_0\leq \|\Delta\|_1=h\leq nq\cdot \|\Delta\|_0=nqi,
    \]
    implying that \[
    \left\lceil\frac{h}{nq}\right\rceil \leq i\leq \min\{h,2nk\}.
    \]
    After selecting $i$ out of $d$ entries to be nonzero, the number of possible $\Delta$ is $2^i\cdot N(h,i,nq)$ \cite{van_Lint_Wilson_2001}, where \begin{align*}
    N(h,i,nq):=\sum_{j=0}^{\min\{i,\lfloor\frac{h-i}{nq}\rfloor\}} (-1)^j\binom{i}{j}\binom{h -jnq-1}{i-1}
    \end{align*}
    is the number of bounded positive integer solutions of \[
    b_1+b_2+\ldots+b_i = h  \,\text{ s.t. }\, b_j\in[nq], 1\leq j\leq i.
    \]
    Here we need to multiply $N(h,i,nq)$ by $2^i$ because the non-zero entries in $\Delta$ is in $\{\pm 1,\ldots,\pm nq\}$. Hence, we can conclude that \begin{equation}
    N_{\max}(D) \leq \sum_{h=0}^{\lfloor nD\rfloor}\sum_{i=\lceil h/nq\rceil}^{\min\{h,2nk\}}\binom{d}{i} \cdot 2^i \cdot N(h,i,nq).
    \label{eq::upper_bd_N_max_D}
    \end{equation}
    Substituting \eqref{eq::lower_bd_U} and \eqref{eq::upper_bd_N_max_D} into \eqref{eq::uniform_u_lb}, we have \begin{align*}
    C(D,0) \geq \log \frac{\sum_{c=0}^{k}\binom{d}{nc}q^{nc}}{\sum_{h=0}^{\lfloor nD\rfloor}\sum_{i=\lceil h/nq\rceil}^{\min\{h,2nk\}}\binom{d}{i} \cdot 2^i \cdot N(h,i,nq)}.
        \end{align*}
    Since the communication cost $C(D,0)$ is non-negative by definition, we obtain the result as stated in this theorem.
\end{proof}

\begin{figure}
    
\begin{subfigure}{\columnwidth}
    \centering
    \includegraphics[width=0.4\linewidth]{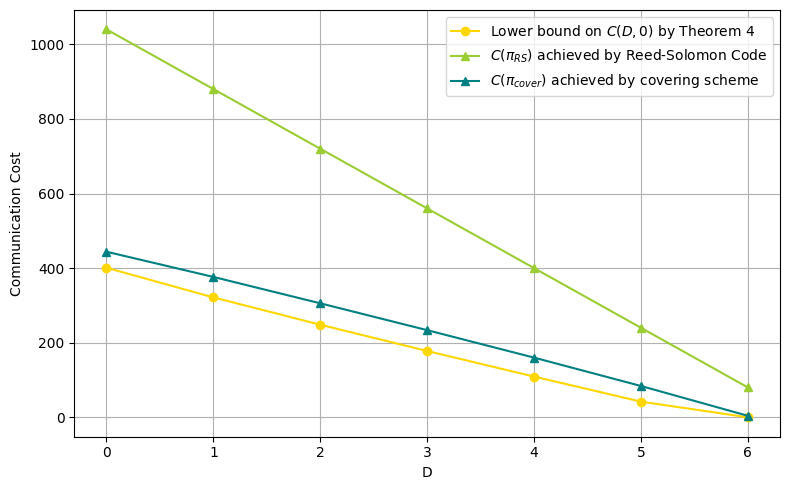}
    \caption{Comparison between the achievable communication cost using the covering scheme, Reed-Solomon code scheme, and the lower bound in Theorem~\ref{thm::converse_approx}, with $\delta=0$, $n=4$, $k=6$, $d=2^{20}$, and $q=1$.}
    \label{fig::result_q_1}
\end{subfigure}
\begin{subfigure}{\columnwidth}
    \centering
    \includegraphics[width=0.4\linewidth]{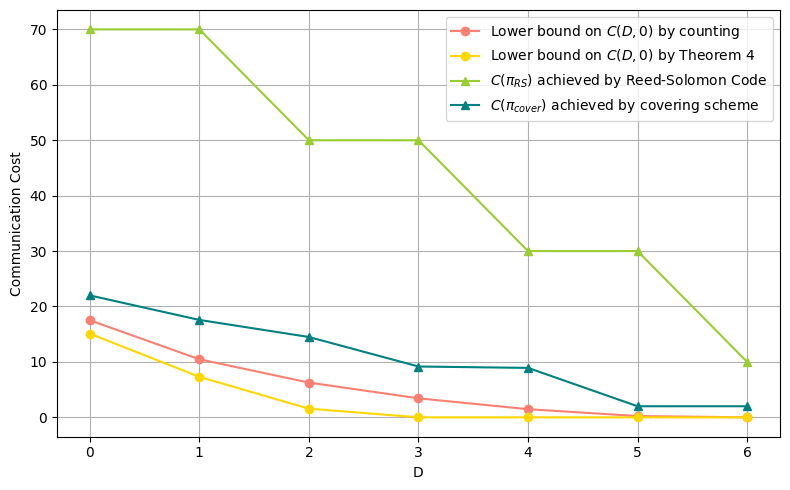}
    \caption{Comparison between the achievable communication cost using the covering scheme, Reed-Solomon code scheme, the lower bound in Theorem~\ref{thm::converse_approx}, and the lower bound in \eqref{eq::chi_2_lb_uniform_u} by counting, with $\delta=0$, $n=2$, $k=3$, $d=11$, and $q=2$.}
    \label{fig::result_q_2_k_3}
\end{subfigure}
        
\caption{Numerical results of the achievable communication costs and the lower bounds.}
\label{fig::delta_0_compare}
\end{figure}

We illustrate the numerical results of this theorem in Fig.~\ref {fig::delta_0_compare}. For the case with $q=1$, which is usually referred to as the \emph{frequency estimation problem}, the gap between the communication cost achieved by the covering scheme and the lower bound in Theorem~\ref{thm::counting_converse} is small. For $q\geq 2$, the performance of the lower bound in \eqref{eq::converse_approx} is worse, and so we also compare the achievable communication costs with the lower bound obtained by counting $|\mathcal{U}|$ and $N_{\max}(D)$, namely \eqref{eq::uniform_u_lb}.

\begin{remark}
    Here, we analyze the gap between the lower bound in \eqref{eq::converse_approx} and the communication cost achieved by the covering scheme when $q=1$. Under the assumption $d\gg nk$, we can approximate the numerator in \eqref{eq::converse_approx} by \[
    \sum_{c=0}^k \binom{d}{nc} q^{nc} = \sum_{c=0}^k \binom{d}{nc} \geq \binom{d}{nk} \approx \frac{d^{nk}}{(nk)!}.
    \]
    Consider the denominator in \eqref{eq::converse_approx}. Since \begin{align*}
    \sum_{h=0}^{\lfloor nD\rfloor} N(h,i,n)=|\{(b_1,\ldots,b_i)\in[n]^i: \sum_{j=1}^i b_j\leq \lfloor nD\rfloor\}|\leq|\{(b_1,\ldots,b_i)\in[n]^i\}|,
    \end{align*}
    we can upper bound $\sum_h N(h,i,n)$ by $n^i$. Therefore, it follows from \eqref{eq::upper_bd_N_max_D} that  \[
    N_{\max}(D)\leq \sum_{i=0}^{\lfloor nD\rfloor}\binom{d}{i}\cdot 2^i\sum_{h=0}^{\lfloor nD\rfloor }N(h,i,n)\leq \sum_{i=0}^{\lfloor nD\rfloor} \binom{d}{i}\cdot (2n)^i.
    \]
    Since the last term with $i=\lfloor nD\rfloor$ will dominate the summation for sufficiently large $d$, we approximate $N_{\max}(D)$ by the last term \[
    \binom{d}{\lfloor nD\rfloor}\cdot (2n)^{\lfloor nD\rfloor}\approx \frac{(2nd)^{\lfloor nD\rfloor}}{\lfloor nD\rfloor!}.
    \]
    Substitute the approximations above to \eqref{eq::converse_approx}, we obtain the lower bound \begin{equation}
    \label{eq::q_1_lb}
    C(D,0) \gtrsim \log \left(\frac{d^{nk}}{d^{\lfloor nD\rfloor}}\cdot \frac{\lfloor nD\rfloor !}{(nk)!\cdot (2n)^{\lfloor nD\rfloor}}\right).
    \end{equation}
    When $d$ is sufficiently large, such as the case in our numerical example in Fig.~\ref{fig::result_q_1}, this lower bound is dominated by the term involving $d$, which is $(nk-\lfloor nD\rfloor)\log d$.

    Recall that, by Corollary~\ref{cor::cover_size}, the communication cost achieved by the covering scheme when $q=1$ satisfies \begin{equation}
    \label{eq::q_1_ub}
    C(D,0) = n\log |\mc Y| \leq n\left(\lceil k-D\rceil \log d + 1\right).
    \end{equation}
    Since both the lower bound in \eqref{eq::q_1_lb} and the upper bound in \eqref{eq::q_1_ub} are approximately equal to $n(k-D)\log d$, the gap between these bounds is very small. This is consistent with the numerical results in Fig.~\ref{fig::result_q_1} for $d=2^{20}$ and $q=1$. Moreover, we can conclude that, \[
    C(D,0) \approx n(k-D)\log d.
    \] 
\end{remark}

\begin{remark}
    One observation from Figs.~\ref{fig::result_q_1} and \ref{fig::result_q_2_k_3} is that, when the distortion $D$ approaches $kq$, the upper bound achieved by the covering scheme and the lower bounds all approach $0$, and become exactly $0$ when $D=kq$.
    
    For the lower bounds, when $D\geq kq$,  we have $N_{\max}(D)=|\mathcal{U}|$ since for any $u,u'\in \mc U$ with $u=\frac{1}{n}\sum_{i=1}^n x_i$ and $u'=\frac{1}{n}\sum_{i=1}^n x_i'$, \[
    \|u-u'\|_1 \leq \frac{1}{n}\sum_{i=1}^n \|x_i-x_i'\|_1 \leq \frac{1}{n}\cdot n\cdot (kq) = kq.
    \]
    Therefore, from \eqref{eq::uniform_u_lb}, we have \[
    C(D,0) \geq \log \frac{|\mc U|}{N_{\max}(D)}  = 0.
    \]

    For the upper bound $C(\pi_{\it cover})$, by substituting $D=kq$ into \eqref{eq::yd_size}, we obtain that $|\mc Y(kq)|=1$, so that \[
    C(D,0)\leq C(\pi_{\it cover}) = n\log |\mc Y(kq)| = n\log 1 = 0.
    \]
    As such, in Figs.~\ref{fig::result_q_1} and \ref{fig::result_q_2_k_3}, the gap(s) between the upper bound achieved by the covering scheme and the lower bounds approach(es) $0$ when $D$ approaches $kq$.
\end{remark}

\section{Conclusion and Discussion}
\label{sec::discussion}
In this paper, we studied lossy compression for aggregating sparse local models in a distributed system. 
With a more communication-efficient covering scheme, we first improved the achievable communication-accuracy trade-off compared with the one using the Reed-Solomon code.
Furthermore, we showed that the communication costs of the two commonly used sketching methods can be reduced by combining them with the covering scheme. 
At the end of this paper, we proved a tighter converse of the communication-accuracy tradeoff by employing the $\mathrm{f}$-divergence. 

One important indication of our numerical result is that, to aggregate sparse local models in a distributed system, properly compressing each local model is important for achieving a communication-accuracy trade-off that is close to the lower bound provided by the converse. A similar observation was made in \cite{chen2023communication} under the privacy and security requirements. The gap between the achieved communication cost by the covering scheme and the lower bound we derived in Theorem~\ref{thm::converse_approx} is small when $q=1$. 
However, this does not hold for a general $q$, and the tighter lower bound obtained through a counting argument is difficult to compute due to the complexity involved. We leave these issues for further investigation in future work.

Another problem is that even though combining the covering scheme and sketching methods can reduce the communication cost when the distortion is above a certain threshold, the impacts of this approach on privacy and security are not studied in this paper. It would be interesting to see how the covering scheme affects the privacy in CS and CMS.

\clearpage
\bibliographystyle{IEEEtran}
\bibliography{ref.bib}


\end{document}